\documentclass[aps,pra,a4paper,twocolumn,superscriptaddress,floatfix]{revtex4}

\usepackage[latin1]{inputenc}
\usepackage{amsthm}
\usepackage{amssymb}
\usepackage{graphicx}
\usepackage{amsmath}
\usepackage{bbold}
\usepackage{bbm}
\usepackage{mathtools}

\usepackage{nonfloat}
\usepackage{color}
\usepackage[pdftex]{hyperref}
\usepackage{braket}
\usepackage{graphicx}
\usepackage{dsfont}

\DeclareGraphicsExtensions{.jpeg, .png, .pdf}

\newtheorem{Def}{Definition}
\newtheorem{thm}{Theorem}
\newtheorem{prop}[thm]{Proposition}
\newtheorem{lemma}[thm]{Lemma}
\newtheorem{cor}[thm]{Corollary}
\newtheorem{ex}{Example}

\begin{document}

\title{Entanglement--Saving Channels}
\author{L. Lami}
\affiliation{\mbox{Universitat Aut\`onoma de Barcelona, ES-01893 Bellaterra (Barcelona), Spain} and Scuola Normale Superiore, I-56126 Pisa, Italy.}

\author{V. Giovannetti}
\affiliation{NEST, Scuola Normale Superiore and Istituto Nanoscienze--CNR, I-56127 Pisa, Italy.}

\begin{abstract}
The set of Entanglement Saving  (ES) quantum channels is introduced and characterized. These are completely positive, trace preserving transformations which when acting locally on a bipartite quantum system initially prepared into a maximally entangled configuration,  preserve its entanglement even when applied an arbitrary number of times. In other words,   a quantum channel $\psi$ is said to be ES if  its powers $\psi^n$ are not entanglement-breaking for all  integers $n$.  We also characterize the properties of the Asymptotic Entanglement Saving (AES) maps. These form a proper subset of the ES channels that is constituted by those maps which, not only preserve entanglement for all finite $n$, but which also sustain an explicitly not null level of entanglement in the asymptotic limit~$n\rightarrow \infty$.  
Structure theorems are provided for ES and for AES maps which  yield an almost complete characterization of the former and a full characterization of the latter. 
\end{abstract}

\maketitle

\section{Introduction}

Entanglement is a distinctive feature of quantum mechanical systems and  the key resource for  quantum data processing~\cite{ENTAN}. Even though ubiquitous this exotic form of correlations is  extremely  difficult to create and preserve, the problem arising from its monogamous character and from the tendency of quantum systems to establish spurious connections with environmental degrees of freedom that are not directly 
under experimental control~\cite{MONO1}. 
A comprehensive study of the processes which tend to deteriorate entanglement can be conducted in the context of quantum channels, i.e. completely positive, trace preserving, linear mapping operating on the set of the density matrices which describe the physical states of a quantum system. According to the postulate of quantum mechanics, quantum channels represent the most general physical transformations which 
a quantum system can undergo when interacting with an external, initially uncorrelated, environment, e.g. see Ref.~\cite{REVIEW,WolfQC,HOLEVOBOOK}.

A very special set of quantum channels is constituted by the so called entanglement--breaking maps~\cite{HorodeckiShorRuskai}. They represent the most detrimental form of noise one has to  face in any experimental  implementation of quantum information processing: when acting locally on an initially entangled  bipartite system, an entanglement--breaking map produces an output that is separable (i.e. not entangled) no matter how intense the initial entanglement was. 
Of course not all
the completely positive, trace preserving transformations 
are entanglement--breaking, the vast majority of maps operating on quantum mechanical system
 representing milder forms of noise. 
In Ref.~\cite{V} a classification of these less disruptive processes was proposed which is based on the accumulated effect arising under iterative applications of a given map. 
In particular a channel $\psi$ was defined to be entanglement--breaking of order $n$  if it requires $n$ recursive applications to  remove all the entanglement initially present in the system. Building up from such approach a series of functionals~\cite{V,LV}  were introduced aimed to quantify  how noisy a given quantum channel is, under the assumption that correcting filtering processes are allowed between two subsequent applications of the latter. 
Aim of the present work is to develop further on these ideas  by focusing  on two very special subsets of quantum channels which we dub Entanglement Saving (ES) and Asymptotically Entanglement Saving (AES), respectively.  An ES channel $\psi$ describes an extremely  weak, yet nontrivial, form of noise which is characterized by the property of  preserving the entanglement of any  bipartite  maximally entanglement state even when applied locally  $n$ times recursively on one of the two subsystems, with $n$ being a generic integer.  Accordingly a map $\psi$ is ES if for all $n$, its iterative $n$--fold application $\psi^n$  is  not entanglement--breaking.
The AES channels constitute a special subclass of the red ES set, formed by those maps which drive (as $n$ goes to infinite) a maximally entangled state of the composite system toward a final configuration which still contains an explicitly not null level of entanglement. 
Structural theorems and a complete characterization of these maps are derived for case of finite dimensional systems. 

Due to the rather technical character of some of the theorems we derive, in writing the present manuscript we made an effort to be as self--contained as possible.
In particular, the first two sections are devoted to review some known facts concerning the theory of quantum channels.
Specifically, in Sec.~\ref{sec generalities} we set the notation and a provide a formal characterization of unitary and entanglement--breaking maps. We also recall the Bloch representation  for qubit channels and the Kadison--Schwarz  inequality which will be extensively used in the subsequent sections. 
While referring to the lecture notes of M. M. Wolf~\cite{WolfQC} as fundamental reference, in Sec.~\ref{sec advanced} instead we give a comprehensive review of the spectral properties of quantum channels. 
Building from this technical introduction we start hence to present our original contributions. In Sec.~\ref{sec UEP} we begin by introducing the notion of universal entanglement--preserving channels, defined as those completely positive, trace--preserving linear maps which preserve all forms of entanglement, no matter how weak it may be at the beginning of the transformation. These can be seen as the counter--parts of entanglement--breaking channels  and, confirming a result 
which is intuitively expected,  we formally prove that they coincide with the set of unitary transformations (see Theorem~\ref{UEP U}). 
The study of ES channels is then presented in Sec.~\ref{sec ES}. In particular, we present a characterization of these maps based on assumption that their determinant is not null~(see Theorem~\ref{ES det =/ 0}) and use this result to give a complete classification  for qubit systems~(see Sec.~\ref{subs ES q}). 
Section~\ref{sec AES} is hence devoted to the study of AES maps. Conclusions and final discussion of the results are then provided in Sec.~\ref{sec con}. 

\section{Basic Theory of Quantum Channels} \label{sec generalities}

\subsection{Generalities}
 In this section  we shall set the notation and review some basic results on the theory of quantum channels
by focusing on the case of finite--dimensional quantum systems.
Recall that a quantum channel provides the proper mathematical framework to describe  the dynamical evolution of an open quantum system. 
 It is well--known  that such a transformation can be equivalently seen as \emph{i)} a unitary interaction with an external ancilla which is later discarded (Stinespring representation), as  \emph{ii)}  a sum of matrix conjugation operations (Kraus representation), or finally as  \emph{iii)} an abstract linear, completely positive, trace--preserving superoperator (axiomatic approach)~\cite{NC,REVIEW,WolfQC,HOLEVOBOOK}.

\subsection{Notation}

In what follows, we will denote by $\mathcal{H}(d;\mathds{C})$ the set of $d\times d$ hermitian matrices. In some cases it can be useful to consider more generally the set of $d\times d$ complex (real) square matrices, denoted by $\mathcal{M}(d;\mathds{C})$ (or $\mathcal{M}(d;\mathds{R})$, respectively). Through the paper, it will be necessary to distinguish between the operator statements $A\geq 0$, $A>0$ and $A\geq 0$ but $\det A=0$ (i.e. $A$ has no negative eigenvalues, but at least one of them is zero). We will refer to the first case by saying simply that $A$ is \emph{positive}, to the second by specifying that $A$ is \emph{strictly positive}, and to the third by using the expression \emph{semipositive}. For instance, we recall that the states of a $d$--dimensional quantum system are represented by positive matrices, either strictly positive or semipositive.

For what concerns the operations between linear subspaces $\mathcal{H}_1,\mathcal{H}_2$, we shall keep the simple notation $\mathcal{H}_1+\mathcal{H}_2$ for the sum, while deserving the notation $\mathcal{H}_1\oplus\mathcal{H}_2$ for \emph{orthogonal} sums, i.e. sums in which we want to specify that the addends are indeed orthogonal. Analogously, $X\oplus Y$ will indicate the direct sum of two operators acting on \emph{orthogonal} spaces. 

Now, let us fix the notation concerning the superoperators. There are various sets of interesting linear maps acting on states of a quantum system (or more generally, on square matrices). We will use a set of convenient abbreviations to indicate them, all written in bold types. In general, block capital letters indicate convex sets and italic capital letters sets defined by nonlinear equations, while small letters impose further linear constraints. Moreover, a subscript $d$ can be added if necessary to specify the dimension of the system (or the size of the square matrices) on which the maps are acting. In our conventions, the letters $\mathbf{P}$, $\mathbf{CP}$, $\mathbf{EB}$, $\boldsymbol{\mathcal{U}}$ denote the sets of positive, completely positive, entanglement--breaking or unitary channels, respectively. The small letters $\mathbf{t}$ and $\mathbf{u}$ are used to specify the trace--preserving or the unital conditions. Thus, with the above rules $\mathbf{CPt}_2$ will be the set of completely positive, trace--preserving qubit channels, while $\boldsymbol{\mathcal{U}}_3$ will denote the set of conjugations by $3\times 3$ unitary matrices on states of a 3--dimensional system. Recall that the matrix transposition, denoted by $T$, is a positive but not completely positive map, that is $T\in\mathbf{Ptu}-\mathbf{CPtu}$. The partial tranpose of a bipartite state $\rho$ only with respect to the second subsystem will be usually denoted by $\rho^{T_B}$.
There is a very natural operation we can define between quantum channels (or more generally between linear maps acting on matrices), i.e. their \emph{composition}. It consists of the consecutive application of two channels (linear maps) $\psi$ and $\phi$. As usual in linear algebra, the simple juxtaposition $\phi\psi$ of the symbols denotes the consecutive application of $\psi$ \emph{firstly}, and of $\phi$ \emph{secondly}. In the same way, $\phi^n$ will indicate the $n$--fold composition of $phi$ with itself. 

We adopt the standard notation $\text{rk}\, \mathcal{L}$ to indicate the rank (i.e. the dimension of the image) of a linear operator (e.g. a square matrix) and $\ker \mathcal{L}$ to denote its kernel (i.e. the linear subspace of vectors $x$ such that $\mathcal{L}x=0$).
Another standard notation is $\sigma(\mathcal{L})$ for the spectrum and $s(\mathcal{L})$ for the set of singular values of $\mathcal{L}$. Naturally, $s_i^\downarrow(\mathcal{L})$ refers to the $i$th greatest singular value of $\mathcal{L}$. Actually, they are understood to be \emph{multisets} rather than simple sets. In a multiset each element can be repeated a number of times equal to its multiplicity. We denote by $a_\mathcal{L}(\lambda)$ and $g_\mathcal{L}(\lambda)$ the algebraic and geometric multiplicities of the eigenvalue $\lambda\in\sigma(\mathcal{L})$, respectively. The definition of singular values requires of course the presence of a positive (semi)definite scalar or hermitian product on the vector space. If the vector space we are dealing with is $\mathds{R}^n$ or $\mathds{C}^n$ we will always choose the standard scalar or hermitian product. Instead, we will regard to the set of square matrices in every dimension as equipped with the Hilbert--Schmidt positive definite hermitian product $\text{Tr}[X^\dag Y]$. Observe that the hermitian conjugation $\phi\rightarrow\phi^\dag$ naturally induced on the superoperators is nothing but the Heisenberg representation of their action, and that $\phi\in\mathbf{Pt}\Leftrightarrow\phi^\dag\in\mathbf{Pu}$.

The Schatten norms of index $1\leq p\leq \infty$ are indicated with $\| \cdot \|_p$ and defined by
\begin{equation} \|\mathcal{L}\|_p\,\equiv\, \left(\ \text{Tr}[(\mathcal{L}^\dag \mathcal{L})^{p/2}]\ \right)^{1/p} \ . \label{Schatten} \end{equation}
One has
\begin{equation} \|\mathcal{L}\|_p\ =\ \left( \sum_i s_i^p(\mathcal{L})\right)^{1/p}\ . \label{Schatten sv} \end{equation}
Observe that $\|\cdot\|_2$ is precisely the norm induced by the Hilbert--Schmidt product. The $p=1$ Schatten norm is known also as trace norm. Furthermore, the natural generalization to the $p=\infty$ case imposes
\begin{equation} \|\mathcal{L}\|_\infty\equiv s_1^\downarrow(\mathcal{L})\ . \label{Schatten infty} \end{equation}

In what follows, $\Ket{\varepsilon}=\frac{1}{\sqrt{d}}\ \sum_{i=1}^d \Ket{i}\otimes\Ket{i}$ will always denote the maximally entangled state of a bipartite system $SS'$, with \mbox{$\dim\mathcal{H}_S=\dim\mathcal{H}_{S'}=d$}. The Choi state associated with a linear map $\phi$ acting on $d\times d$ square matrices is by definition
\begin{equation} R_\phi\equiv(\phi\otimes I)(\Ket{\varepsilon}\!\!\Bra{\varepsilon})\ . \label{Choi state} \end{equation}
Finally, we will indicate with $\mathcal{S}_{AB}$ the set of separable density matrices on a bipartite system $AB$ (the subscript can be removed if there is no ambiguity).

\subsection{Kadison--Schwarz inequality}
An important classical inequality that will turn out to be very important is called \emph{Kadison's inequality} (after its discoverer). Its stronger form which exploits the complete positivity assumption is known as \emph{Schwarz inequality}. For the original proof of the Kadison's inequality we refer to~\cite{Kadison}; otherwise, a more intuitive argument can be found in~\cite{Kadison Woronowicz}. The text~\cite{Bhatia} provides a comprehensive reference on the subject. \\

\begin{thm}[Kadison--Schwarz Inequality] \label{Kadison} $\\$
Let $\zeta\in\mathbf{Pu}_d$ be a positive unital map. Then
\begin{equation} \forall\ X=X^\dag\in\mathcal{H}(d;\mathds{C})\ ,\qquad \zeta(X)^2 \leq \zeta(X^2)\ .\label{Kadison P} \end{equation}
Moreover, if $\chi\in\mathbf{CPu}_d$ is completely positive and unital, then
\begin{equation} \forall\ X\in\mathcal{M}(d;\mathds{C})\ ,\qquad \zeta(X)^\dag\zeta(X)\, \leq\, \zeta(X^\dag X)\ . \label{Kadison CP} \end{equation}
\end{thm}

\subsection{Characterizing unitary evolutions}

A peculiar class of quantum channels is formed by the unitary evolutions, acting as $\,\mathcal{U}(\cdot)=U(\cdot)U^\dag$ for some unitary matrix $U$. It is very useful for what follows to have several less direct characterizations of the unitary channels, that is, several \emph{sufficient} conditions to claim that a given channel is indeed unitary. What this kind of results have in common is that for their proof a deep fact known as Wigner's theorem turns out to be fundamental. For the original work we refer the reader to~\cite{Wigneroriginal}, p. 251--254. Otherwise, a direct and mathematically clear proof can be found in~\cite{Wignerdirectproof}.

\begin{thm}[Wigner's Theorem] \label{Wigner} $\\$
Let $T:\mathcal{H}\rightarrow\mathcal{H}$ be a (not necessarily linear) operator on a (not necessarily finite--dimensional) Hilbert space $\mathcal{H}$. Suppose that
\begin{equation} \left| \Braket{T(x)|T(y)} \right|\,\equiv\,\left| \Braket{x|y}\right|\quad\forall\, x,y\in\mathcal{H}\, . \label{Wigner hyp}\end{equation}
Then there exists a real function $\varphi:\mathcal{H}\rightarrow\mathds{R}$ such that
\begin{equation} T(x)\,\equiv\, e^{i\varphi(x)}Vx\, , \label{Wigner thesis}\end{equation}
where $V:\mathcal{H}\rightarrow\mathcal{H}$ is an isometry or an anti--isometry. In particular, if $\mathcal{H}$ is finite--dimensional then $V$ is unitary or anti--unitary.
\end{thm} 

Now, let us present the main result about the alternative characterizations of unitary evolutions. Its proof can be found through Chap. 6 of~\cite{WolfQC}.

\begin{thm}[Alternative Characterizations of Unitary Channels] \label{unitary} $\\$
Let $\phi\in\mathbf{CPt}$ be a quantum channel. Then the following are equivalent:
\begin{enumerate}
\item $\phi$ is a unitary evolution;
\item $\det\phi=\pm 1$, where the determinant is defined as the product of the eigenvalues (see Subsection~\ref{spectrum});
\item there exists the inverse $\phi^{-1}$ of $\phi$, and it is again a quantum channel;
\item $\phi$ maps pure states into pure states, and does not have the form \mbox{$\phi(X)=\Ket{\alpha}\!\!\Bra{\alpha}\,\text{\emph{Tr}}\,X$} for some fixed pure state $\Ket{\alpha}$.
\end{enumerate}
\end{thm}

A few comments are appropriate. In Subsection~\ref{spectrum} we will see that $|\det\phi|\leq 1$, so that the unitary channels are the ones and the only ones exhibiting the \emph{extremal} value of the modulus of the determinant. The third condition of Theorem~\ref{unitary} is usually interpreted as the \emph{irreversibility} of the evolution of an open quantum system. In Section~\ref{sec UEP} we will see another remarkable characterization of the unitary channels.

\subsection{Entanglement--breaking channels}\label{sec eb}

A class of quantum channels that will play a central role in what follows is composed of the so--called \emph{entanglement--breaking channels}~\cite{HorodeckiShorRuskai}. Recall that a quantum channel $\phi\in\mathbf{CPt}$ acting on a system $A$ is called entanglement--breaking (and we will write $\phi\in\mathbf{EBt}$) if for each system $B$ and for each global input state $\rho_{AB}$ of $AB$, the output $(\phi\otimes I)(\rho_{AB})$ is separable. Remarkably,
\begin{equation} \phi\in\mathbf{EB}\,,\ \psi\in\mathbf{CP}\quad\Rightarrow\quad \phi\psi,\,\psi\phi\in\mathbf{EB}\ . \label{EB propagates}\end{equation}
Moreover, it turns out that also $\mathbf{EB}$ (just like $\mathbf{P}$ or $\mathbf{CP}$) is a closed convex set which is in addition closed under the operation of taking the hermitian adjoint:
\begin{equation} \phi\in\mathbf{EB}\quad\Leftrightarrow\quad\phi^\dag\in\mathbf{EB}\ . \label{EB adjoint} \end{equation}
This equivalence follows also from Theorem~\ref{str thm EB}, which we state in a moment.

The problem of the operational characterization of the EB class has been solved in~\cite{HorodeckiShorRuskai}. This solution is the content of the following theorem.

\begin{thm}[Structure Theorem for EB Channels] \label{str thm EB} $\\$
Let $\phi\in\mathbf{EBt}$ be a quantum channel. Then the following facts are equivalent:
\begin{enumerate}
\item $\phi$ is entanglement--breaking.
\item The associated Choi state $R_\phi$ (see~\eqref{Choi state}) is separable.
\item $\phi$ can be written in the \emph{Holevo form} introduced in~\cite{Holevoform}, i.e. there are a (finite) set of density matrices $\{ \rho_i \}$ and positive operators $\{ E_i \}$ satisfying the sum rule $\sum_i E_i=\mathds{1}$, such that
\begin{equation} \phi(X)\ \equiv\ \sum_i \rho_i\, \text{\emph{Tr}}[E_i X] \qquad \forall\ X\in\mathcal{M}(d;\mathds{C})\ . \label{Holevo form} \end{equation}
\end{enumerate}
\end{thm}

Observe that in the Holevo form~\eqref{Holevo form} we can freely suppose that the $\rho_i$ are pure states and that the $E_i$ are (positive) multiples of pure states. This can be seen directly by diagonalizing both operators, and it is also a by--product of the proof. Equation~\eqref{Holevo form} provides exactly the operative interpretation we were looking for. Indeed, it states that the entanglement--breaking channels are exactly those channel which can be implemented by a measurement process (POVM) followed by a re--preparation of the system.

\subsection{Qubit channels}

The qubit case, i.e. the case in which our channels act on a $2$--dimensional system, deserves particular attention. Let us recap the main geometrical tools which become available in this particular framework. It is well--known that a qubit state can be written in the \emph{Bloch representation} as
\begin{equation} \rho=\frac{\mathds{1}+\vec{r}\cdot\vec{\sigma}}{2}\quad , \label{Bloch q} \end{equation}
where $\vec{\sigma}=(X,Y,Z)$ is simply the vector of Pauli matrices, and $|\vec{r}|\leq 1$. Observe that here the pure states are exactly those states $\rho$ whose associated vector $\vec{r}$ has unit modulus. This can be immediately seen by noting that the spectrum of $\mathds{1}+\vec{r}\cdot\vec{\sigma}$ is given by
\begin{equation} \sigma\,(\mathds{1}+\vec{r}\cdot\vec{\sigma})\, =\, \{\, 1+|\vec{r}|,\, 1-|\vec{r}| \, \}\, . \label{Pauli spectrum} \end{equation}
We can choose to represent a quantum qubit channel $\phi$ by means of its action on~\eqref{Bloch q}. This means that $\phi$ is completely specified once we assign the $3\times 3$ real matrix $M$ and the $3$--vector $c$ such that
\begin{equation} \phi\left(\frac{\mathds{1}+\vec{r}\cdot\vec{\sigma}}{2}\right) \ = \ \frac{\mathds{1}+(M\vec{r}+\vec{c})\cdot\vec{\sigma}}{2}\quad . \label{Bloch act} \end{equation}

In view of~\eqref{Bloch act}, we will sometimes indicate the channel $\phi$ with the notation $(M,c)$. Remind that in this picture the unitary evolution $\mathcal{U}(\cdot)=U(\cdot)U^\dag$, where $U=e^{-i\,\vec{\theta}\cdot\vec{\sigma}/2}$ is a SU(2) matrix, is represented by the counterclockwise rotation $R(\vec{\theta})$ of an angle $\theta$ around $\vec{\theta}/|\vec{\theta}|$.

Since in the (orthogonal) Pauli basis $\mathds{1},X,Y,Z$ the linear map $\phi$ is represented by $\left(\begin{smallmatrix} 1 & 0 \\ c & M \end{smallmatrix}\right)$, the spectrum of $\phi$ as a linear application (see Subsection~\ref{spectrum}) is simply given by
\begin{equation} \sigma(\phi)\, =\, \{\,1\,\}\, \cup\, \sigma(M)\, , \label{spectrum q} \end{equation}
and its determinant by
\begin{equation} \det\phi\, =\, \det M\, . \label{det q} \end{equation}

The Bloch representation~\eqref{Bloch act} of the qubit quantum channels allows us to find a useful \emph{canonical decomposition} for this special case. As pointed out firstly in~\cite{KingRuskai}, by applying unitary evolutions to the left and to the right of $\phi=(M,c)\in\mathbf{CPt}_2$ the best \emph{special singular value decomposition} we can achieve has the form $M=P L Q$, with $P,Q\in\text{SO}(3)$ and
\begin{multline} L\ =\ \begin{pmatrix} l_1(M) & 0 & 0 \\ 0 & l_2(M) & 0 \\ 0 & 0 & l_3(M) \end{pmatrix}\ \equiv\\
\equiv\ \begin{pmatrix} s_1(M) & 0 & 0 \\ 0 & s_2(M) & 0 \\ 0 & 0 & \text{sgn} \det (M)\ s_3(M) \end{pmatrix}\ . \label{L matrix} \end{multline}
Here the symbol $\text{sgn}$ denotes the \emph{sign function}, defined by
\begin{equation*} \text{sgn}\, x\, \equiv\, \left\{ \begin{array}{cr} +1 & \text{ if $x>0$}\\ 0 & \text{ if $x=0$}\\ -1 & \text{ if $x<0$} \end{array} \right.\ , \end{equation*}
and $s_i(M)$ indicates the $i$th singular value of $M$. Usually one can suppose $|s_3(M)|\leq s_1(M),s_2(M)$, so that $l_1(M),l_2(M)\geq 0$ and only $l_3(M)$, which has the lowest modulus, can be negative. These $l(M)$ are called \emph{special singular values} of the real $3\times 3$ matrix $M$. Once the decomposition $M=P L Q$ is obtained, we can define $t\equiv P^T c$ and write 
\begin{equation} \phi\ =\ (M,c)\ =\ P\ (L, t)\ Q\ =\ \mathcal{U}\ \Lambda\ \mathcal{V}\ . \label{q canonical form} \end{equation}
Here $\mathcal{U},\mathcal{V}$ are the unitary channels corresponding to $P,Q\in\text{SO}(3)$, and $\Lambda\equiv(L,t)$ is the \emph{canonical diagonal form} of $\phi$. Remarkably, since the unitary evolutions are one--to--one applications between density matrices, the positivity, the complete positivity and the entanglement--breaking conditions are not affected if one passes to the canonical diagonal form. That is, with the notations of~\eqref{q canonical form} we have
\begin{gather} \phi\in\mathbf{Pt}_2\quad\Leftrightarrow\quad\Lambda\in\mathbf{Pt}_2\ , \label{P q canonical form} \\ \phi\in\mathbf{CPt}_2\quad\Leftrightarrow\quad\Lambda\in\mathbf{CPt}_2\ , \label{CP q canonical form} \\
\phi\in\mathbf{EBt}_2\quad\Leftrightarrow\quad\Lambda\in\mathbf{EBt}_2\ . \label{EB q canonical form} \end{gather}

What can be said about the entanglement--breaking qubit channels? As a matter of fact, it turns out that the Bloch representation allows us to find new useful characterizations for the EB qubit channels. Let us recall the following theorem, which is obtained by joining together Theorem~1 and~2 of~\cite{RuskaiEBqubit}. As usual, $\rho^{T_B}$ denotes the partial transpose of the bipartite state $\rho$ with respect to the second subsystem.

\begin{thm}[EB Conditions for Qubit Channels] \label{EB q} $\\$
Let $\phi\in\mathbf{CPt}_2$ be a qubit channel. Then the following facts are equivalent:
\begin{enumerate}

\item $\phi$ is entanglement--breaking.

\item $R_\phi^{T_B}\geq 0$ .

\item $T\phi\in\mathbf{CPt}_2\ \text{or} \ \phi T\in\mathbf{CPt}_2$ ($\,T$ is the matrix transposition channel).

\item $\phi$ has the ``sign--change'' property that changing any $l_i\mapsto -l_i$ of the matrix $L$ defined in~\eqref{L matrix} and employed in the canonical diagonal decomposition~\eqref{q canonical form} yields another completely positive map.

\item $\|\,R_\phi\,\|_\infty\,\leq\,\frac{1}{2}$ .

\end{enumerate} 
\end{thm}

\section{Advanced Theory of Quantum Channels} \label{sec advanced}

Through this section, we will review some more advanced topics in the theory of quantum channels. The results we will present have been already studied in the literature, and a set of references to the previous works will be provided. However, we find convenient for the sake of clarity to uniform the notation and group together useful facts to be used extensively in the rest of the paper.

\subsection{Spectral properties of positive maps} \label{spectrum}
First of all, let us review the main results in the study of the spectral properties of positive, trace--preserving maps. Recall that a map $\phi\in\mathbf{Pt}$ is first of all a linear operator acting on the real space $\mathcal{H}(d;\mathds{C})$ of $d\times d$ hermitian matrices. Like all the linear operations on a $d^2$--dimensional real space, also $\phi$ can be regarded as a $d^2\times d^2$ real matrix. Therefore, a \emph{spectrum} $\sigma(\phi)$, the related \emph{eigenvectors} (actually, we should say eigen\emph{matrices}!) and the whole \emph{Jordan form} (see Chap. 3 of~\cite{HJ1} for an introduction to this standard subject) can be naturally associated to it.
Let us discuss some general properties concerning the spectrum of an arbitrary $\mathbf{Pt}$ map. The condition of complete positivity (pertaining to the physical quantum channels) has to be regarded as a particular case. The knowledge of these basic properties will be very useful through the following chapters. For an excellent overview with all the proofs, we refer the reader to Chap. 6 of~\cite{WolfQC}. We will approximately follow this text for our exposition. For the sake of simplicity, let us group all together in a theorem. \\

\begin{thm}[Spectral Properties of $\mathbf{Pt}$ Maps] \label{spect prop Pt} $\\$
Let $\phi\in\mathbf{Pt}$ be a positive, trace--preserving map with spectrum $\sigma(\phi)$ (counting multiplicities). Then the following properties hold.
\begin{enumerate}

\item The eigenvalues are real or come in complex conjugate pairs $z,z^*$, with the same multiplicity and Jordan structure for $z$ and $z^*$. If $\lambda\in\sigma(\phi)$ is real then the related eigenvector can be chosen hermitian. Otherwise, $\phi(Z)=z Z\Leftrightarrow\phi(Z^\dag)=z^* Z^\dag$. As a consequence, the linear span of the eigenvectors pertaining to complex conjugated eigenvalues is a real subspace, i.e. it admits a basis composed of hermitian operators. Finally, the trace--preserving condition imposes that the eigenvectors associated with $1\neq\lambda\in\sigma(\phi)$ can be chosen traceless.

\item Let $X=\phi(X)$ be an hermitian fixed point of $\phi$. Denote by \mbox{$X = X_+ - X_-$} the decomposition of $X$ into its positive and negative spectral parts \mbox{$X_\pm\geq 0$}. Then $X_+,X_-$ as well as $|X| = X_+ + X_-$ are (positive definite) fixed points of $\phi$.

\item There exists at least a density matrix $\rho_0\geq 0$ which is fixed by $\phi$ (that is, \mbox{$\phi(\rho_0)=\rho_0$}).

\item All the eigenvalues lie in the complex unit circle (i.e. \mbox{$\lambda\in\sigma(\phi)\Rightarrow|\lambda|\leq 1$}). In particular, the determinant of the channel satisfies $|\det\phi|\leq 1$. Moreover, the eigenvalues with modulus equal to $1$ can only have trivial Jordan blocks (this is the same as to say that the algebraic and geometric multiplicities are always the same for eigenvalues with unit modulus, or that $d_k=1$ in~\eqref{Jordan}, and so no nilpotent operator is indeed present on these generalized eigenspaces).

\item Let
\begingroup
\renewcommand*{\arraystretch}{1.8}
\begin{equation} \begin{array}{c} \phi=\sum_k (\lambda_k P_k + N_k)\, , \\
P_h P_k = \delta_{hk} P_k, \quad \text{\emph{Tr}} P_k = d_k, \quad \sum_k P_k = \mathds{1}, \\
N_k^{d_k}=0, \qquad P_k N_k = N_k P_k = N_k \end{array} \label{Jordan} \end{equation}
\endgroup
be the Jordan decomposition for $\phi$. In~\eqref{Jordan} the $\lambda_k$ are the eigenvalues, the $P_k$ the (not necessarily orthogonal!) projectors onto the generalized subspaces, and the $N_k$ are nilpotent (super)operators. Then the following combinations of spectral projectors are all $\mathbf{Pt}$ maps. Moreover, if $\phi\in\mathbf{CPt}$, then so are these maps.

\begin{gather}
E_\phi\ \equiv\, \sum_{k:\, |\lambda_k|=1} P_k\ , \label{E phi} \\
I_\phi\ \equiv\, \sum_{k:\, |\lambda_k|=1} \lambda_k^*\, P_k\ , \label{I phi} \\
\phi_\infty\ \equiv\, \sum_{k:\, \lambda_k=1} P_k\ . \label{phi infty}
\end{gather}

\end{enumerate}
\end{thm}

\begin{proof}
\begin{enumerate}
\item[]
\item These are all well--known consequences of the hermiticity--preserving condition, which implies that $\phi$ is indeed a real endomorphism of the real vector space $\mathcal{H}(d;\mathds{C})$, that is, a real $d^2\times d^2$ matrix. The trace--preserving condition, moreover, implies that 
\begin{equation*} \phi(Z)=\lambda Z,\ \lambda\neq 1\ \Rightarrow\ \text{Tr}[Z] = \text{Tr}[\phi(Z)] = \lambda \text{Tr} [Z] = 0 \end{equation*}

\item Call $P_+$ the projector onto the positive part of $X$, that is $P_+ X P_+ = X_+$. Since we know that \mbox{$X=X_+ - X_- = \phi(X)$}, we can write
\begin{multline*} X_+\, =\, P_+ X P_+\, =\, P_+ \phi(X_+ - X_-) P_+\, =\\
=\, P_+ \phi(X_+) P_+ - P_+ \phi(X_-) P_+\, \leq\\
\leq\, P_+ \phi(X_+) P_+\quad \Rightarrow\quad P_+ \phi(X_+) P_+ - X_+\, \geq\, 0\, . \end{multline*}
We would like to prove that indeed \mbox{$P_+ \phi(X_+) P_+ - X_+ = 0$}. Thanks to the fact that \mbox{$P_+ \phi(X_+) P_+ - X_+ \geq 0$}, it will suffice to check that the trace of this operator is not greater than zero.
\begin{multline*} \text{Tr}\, [P_+ \phi(X_+) P_+ - X_+]\, =\\
=\, \text{Tr}\, [P_+ \phi(X_+)] - \text{Tr}[X_+]\, \leq\\
\leq\, \text{Tr}\, [\phi(X_+)] - \text{Tr}[X_+]\, =\, 0\, . \end{multline*}

\item Since $\phi^\dag\in\mathbf{Pu}$, we know that $\phi^\dag(\mathds{1})=\mathds{1}$, and so $1\in\sigma(\phi^\dag)$. It is well--known that the spectrum of the hermitian conjugate matrix is nothing but the complex conjugate of the original spectrum (with the same multiplicities). As a consequence, $1\in\sigma(\phi)$, and moreover the previous point ensures that the corresponding eigenvector can be chosen positive (and of unit trace, of course).  

\item These nontrivial facts descend from the observation that $\mathbf{Pt}$ (or $\mathbf{CPt}$) are \emph{compact} sets closed for composition, and an eigenvalue with modulus grater than $1$ produces unbounded powers. Moreover, even a nontrivial Jordan block pertaining to an eigenvalue with modulus $1$ has unbounded powers, as can be easily verified by direct calculations.

\item The proof of this statement relies on exploiting a result of number theory known as Dirichlet's theorem on simultaneous Diophantine approximations to show that both $E_\phi$ and $I_\phi$ are limit points of the sequence of powers $(\phi^n)_{n\in\mathds{N}}$, and so that they must be $\mathbf{Pt}$ (or $\mathbf{CPt}$) if so is $\phi$. Moreover, one can easily see that $\phi_\infty$ is the limit of the \emph{means of the powers}, that is
\begin{equation*} \phi_\infty\, =\, \lim_{n\rightarrow\infty} \frac{1}{n}\, \sum_{i=1}^n \phi^i\ , \end{equation*}
and therefore must be again positive (or completely positive) and trace--preserving, since these properties are preserved under compositions, convex combinations and limits.

\end{enumerate}
\end{proof}

Let us fix some nomenclature and notation. In view also of Theorem~\ref{spect prop Pt}, an eigenvalue of unit modulus of a positive, trace--preserving map $\phi$ is called a \emph{peripheral eigenvalue}, belonging to the \emph{peripheral spectrum} $\sigma_P(\phi)$. An eigenvector pertaining to a peripheral eigenvalue is called a \emph{phase point} of the map. It is nothing but a square matrix $Z$ such that $\phi(Z)=e^{i\theta} Z$ for some real number $\theta$. The linear span of the phase points, denoted by $\chi_\phi$, is called \emph{phase subspace}. Also the \emph{fixed subspace}, that is the eigenspace of $\phi$ pertaining to the eigenvalue $1$ (or the set of \emph{fixed points}), deserves a special notation, being indicated by $\eta_\phi$.

Now, the following central questions naturally arise: \emph{what is the most general structure of the phase subspace of a quantum channel? And what is the most general action of $\phi$ on this phase subspace?} The rest of this section is devoted to present the answer to this question. The proofs of the central claims are highly nontrivial and quite technical, even if the claims themselves can be written in a quite simple way. Consequently, we shall break the general argument into several smaller constructions, identified by the various subsections.

The first step can be immediately understood, thanks to Theorem~\ref{spect prop Pt}. Suppose that we want identify the structure of the phase subspace $\chi_\phi$ of a completely positive map $\phi\in\mathbf{CPt}$. We can then construct the spectral projection $E_\phi$ as defined in~\eqref{E phi}, and observe that:
\begin{itemize}
\item $E_\phi\in\mathbf{CPt}$;
\item $\chi_\phi = \eta_{E_\phi}$;
\item $E_\phi$ is \emph{idempotent}, that is $E_\phi^2 = E_\phi$.
\end{itemize}
Thanks to this simplification, from now on we can restrict ourselves to study only the sets of \emph{fixed points} of \emph{idempotent} completely positive maps. This will be enough to understand the structure of the phase subspace of every completely positive channel.

\subsection{Restricting to maps with a strictly positive definite fixed point} \label{psibar}

We know from Theorem~\ref{spect prop Pt} that every positive map has a positive fixed point $\rho_0$. Of course, there are maps for which $\rho_0 > 0$ is \emph{strictly} positive definite, and other maps for which $\rho_0$ is still positive but has some zero eigenvalue.
Even if it is not a priori obvious, the entire theory of the phase points (or of the fixed points) becomes much simpler is we would allowed to make the assumption that $\rho_0 > 0$. In this subsection we will describe a theoretical construction that allows us to associate with a generic map $\psi\in\mathbf{Pt}$ another map $\tilde{\psi}\in\mathbf{Pt}$ such that $\eta_\psi=\eta_{\tilde{\psi}}$. Let us begin with a little Lemma. Recall that the \emph{support} of an hermitian operator is by definition the subspace spanned by its eigenvectors pertaining to nonzero eigenvalues (i.e. the orthogonal complement of the kernel).

\begin{lemma} \label{supp fix point} $\\$
Let $\psi\in\mathbf{Pt}$, and define the subspace $\mathcal{K}$ as the sum over the positive fixed points of $\psi$
\begin{equation} \mathcal{K}\ = \sum_{\psi(A)\, =\, A\, \geq\, 0} \text{\emph{supp}}\, A\ . \label{eq K} \end{equation}
Then the following statements hold.
\begin{enumerate}

\item The support of every fixed point of $\psi$ is contained in $\mathcal{K}$. Moreover, $\mathcal{K}$ is the smallest subspace enjoying this property.

\item There exists a positive fixed point $\rho_0\geq 0$ of $\psi$ such that $\text{\emph{supp}}\,\rho_0=\mathcal{K}$. A possible choice is for instance $\rho_0=\psi_\infty(\mathds{1})$ (see~\eqref{phi infty}).

\end{enumerate}
\end{lemma}

\begin{proof}
\begin{enumerate}
\item[]
\item By Theorem~\ref{spect prop Pt}, if $X=\psi(X)$ is a fixed point, so is $|X|$. As a consequence, \mbox{$\text{supp}\, X = \text{supp}\, |X| \subseteq \mathcal{K}$} (because $\text{supp}\, |X|$ is an addend of the sum~\eqref{eq K}). That $\mathcal{K}$ is the smallest subspace containing all the supports of the fixed points follows easily if we prove that there exists a (positive) fixed point $\rho_0$ such that $\text{supp}\, \rho_0 = \mathcal{K}$, which is the next claim.

\item If $\mathcal{K}=\text{supp}\, X_1 + \ldots + \text{supp}\, X_n$, we can choose $\rho_0 \equiv |X_1|+\ldots +|X_n|$. Another legitimate choice is $\rho_0=\psi_\infty(\mathds{1})$ because on one hand $\psi\psi_\infty=\psi_\infty$ by the very definition~\eqref{phi infty} (so that $\psi_\infty(\mathds{1})$ is indeed a fixed point), and on the other hand for all $\psi(A)=A\geq 0$ we have also $\psi_\infty(A)=A$, and so \mbox{$A=\psi_\infty(A)\leq\psi_\infty(\|A\|_\infty \mathds{1})=\|A\|_\infty \psi_\infty (\mathds{1})$}, which in turn implies $\text{supp}\, A\subseteq\text{supp}\, \psi_\infty(\mathds{1})$.

\end{enumerate}
\end{proof}

Now that we have constructed this subspace $\mathcal{K}$, we show why it is indeed interesting.

\begin{prop} \label{prop psi bar} $\\$
Let $\psi\in\mathbf{Pt}$ be a positive, trace--preserving map, and define as above the subspace \mbox{$\mathcal{K}= \sum_{\psi(A)\, =\, A\, \geq\, 0} \text{\emph{supp}}\, A$}. Then for all hermitian \mbox{$X=X^\dag$} such that $\text{\emph{supp}}\, X\subseteq \mathcal{K}$, we have also $\text{\emph{supp}}\, \psi(X)\subseteq \mathcal{K}$.
\end{prop}

\begin{proof}
Up to decomposing $X$ into its positive and negative spectral parts, it will suffice to prove the thesis for positive $X$. Use Lemma~\ref{supp fix point} to construct $\rho_0\geq 0$ such that $\text{supp}\, \rho_0 = \mathcal{K}$. Then the hypothesis implies there must be a number $m\in\mathds{R}$ such that $0\leq X\leq m \rho_0$, so that \mbox{$0\leq\psi(X)\leq m \psi(\rho_0) = m \rho_0$}, and this is possible only if \mbox{$\text{supp}\, \psi(X)\subseteq \mathcal{K}$}.
\end{proof}

Then, consider a generic $\psi\in\mathbf{Pt}_d$, and construct the associated subspace $\mathcal{K}$ such that $\dim \mathcal{K}=r\leq d$ and \mbox{$\mathds{C}^d=\mathcal{K}\oplus \mathcal{K}^\perp$}. Thanks to Proposition~\ref{prop psi bar}, it makes sense to define the restriction $\tilde{\psi}:\mathcal{M}(r;\mathds{C})\rightarrow\mathcal{M}(r;\mathds{C})$ by
\begin{equation} X=x\oplus 0\quad \Rightarrow\quad \psi(X)= \tilde{\psi}(x) \oplus 0\, , \label{psi bar}\end{equation}
where all the block decompositions are understood to be in accordance with the space decomposition \mbox{$\mathds{C}^d=\mathcal{K}\oplus\mathcal{K}^\perp$}. This new map $\tilde{\psi}$ enjoys the following properties.
\begin{itemize}

\item $\psi\in\mathbf{Pt}_r$, and moreover if $\psi\in\mathbf{CPt}_d$ then also $\tilde{\psi}\in\mathbf{CPt}_r$.

\item $\eta_\psi = \eta_{\tilde{\psi}}\oplus 0$, once again accordingly with the decomposition $\mathds{C}^d=\mathcal{K}\oplus \mathcal{K}^\perp$.


\item $\tilde{\psi}$ admits a \emph{strictly} positive definite fixed point $\rho_0\in\mathcal{H}(r;\mathds{C})$, which is of course nothing but the restriction of the maximal fixed point provided by Lemma~\ref{supp fix point} to $\mathcal{K}$ (and will be again indicated with $\rho_0$). Let us clarify this tricky point. The positive matrix $\rho_0$ is not invertible on the whole space $\mathds{C}^d$. However, the channel $\tilde{\psi}$ has been deliberately constructed as a restriction to the precise subspace on which $\rho_0$ is indeed invertible.

\end{itemize}

Thanks to the above construction and to the fact that $\eta_\psi=\eta_{\tilde{\psi}}\oplus 0$, we can now assume in our study of the fixed subspace that $\psi$ has a strictly positive definite fixed point. We shall see that this will guarantee the existence of a simpler structure.

\subsection{Theory for unital (idempotent) maps whose hermitian adjoint has a strictly positive definite fixed point} \label{part theory}

We already know that the trace--preserving condition is the hermitian adjoint dual of the unital condition, that is \mbox{$\psi\in(\mathbf{C})\mathbf{Pt}\Leftrightarrow\psi^\dag\in(\mathbf{C})\mathbf{Pu}$} (remind that the hermitian adjoint is taken with respect to the Hilbert--Schmidt hermitian product between matrices). Through this section we will develop the theory of the fixed subspace for \emph{unital} maps, exploiting also the assumptions of idempotence and of invertibility of maximal fixed point of the adjoint. 

A good reason for studying the unital maps instead of their trace--preserving counterpart (hermitian adjoint) comes from the fact that the Kadison's inequalities contained in Theorem~\ref{Kadison} are formulated for unital rather than for trace--preserving maps. Observe that we can not say a priori that the fixed subspaces $\eta_\psi$ and $\eta_{\psi^\dag}$ are related in a particular way. Indeed, they will be in general \emph{different} subspaces. However, they must satisfy some relations. For instance, since we know that the spectrum $\psi^\dag$ is simply the complex conjugated of that of $\psi$, we can a priori say that $1$ must belong to $\sigma(\psi^\dag)$ with the same multiplicity of $1\in\sigma(\psi)$, i.e. that
\begin{equation} \dim\, \eta_\psi\ =\ \dim\, \eta_{\psi^\dag}\ . \label{adjoint fix} \end{equation}

The crucial observation, that naturally leads to the final classification theorem, has been done by Lindblad in~\cite{Lindblad}. We will present it through the following theorem, whose central claim is somewhat a priori unexpected and quite surprising. Here is the point in which the assumption of complete positivity rather than of simple positivity becomes fundamental. Indeed, the inequality~\eqref{Kadison CP}, which is stronger than~\eqref{Kadison P}, will play a decisive role.

\begin{thm}[Lindblad's Theorem] \label{Lind} $\\$
Let $\zeta\in\mathbf{CPu}$ be a completely positive, unital map such that its hermitian adjoint $\zeta^\dag\in\mathbf{CPt}$ has a strictly positive definite fixed point. Then the fixed subspace $\eta_\zeta$ is closed under matrix multiplication.
\end{thm}

\begin{proof}
It suffices to demonstrate that if $Z\in\eta_\zeta$ then $Z^\dag Z\in\eta_\zeta$. Then indeed for all $X,Y\in\eta_\zeta$ one has $(X^\dag+Y)^\dag (X^\dag+Y)\in\eta_\zeta$, and so subtracting $X X^\dag + Y^\dag Y\in\eta_\zeta$ also \mbox{$X Y + Y^\dag X^\dag\in\eta_\zeta$}. Moreover, \mbox{$(X^\dag+iY)^\dag (X^\dag+iY)\in\eta_\zeta$}, and subtracting again $X X^\dag + Y^\dag Y\in\eta_\zeta$ we obtain $X Y - Y^\dag X^\dag\in\eta_\zeta$. Finally, summing with the previous identity yields $X Y\in\eta_\zeta$.

Therefore, let us prove that \mbox{$Z\in\eta_\zeta\Rightarrow Z^\dag Z\in\eta_\zeta$}. Applying~\eqref{Kadison CP} gives
\begin{equation} \zeta(Z^\dag Z) - Z^\dag Z\geq 0\ . \label{Lind eq1} \end{equation}
We would like to prove that indeed the operator on the left--hand side of~\eqref{Lind eq1} is zero. Take the strictly positive definite fixed point of $\zeta^\dag$, namely $\zeta^\dag(\rho_0)=\rho_0>0$, whose existence is guaranteed by hypothesis, and note that
\begin{multline} \text{Tr}\,[\,\rho_0\ \left(\zeta(Z^\dag Z) - Z^\dag Z\right)]\, =\\
=\, \text{Tr}\,[\,\rho_0\ \zeta(Z^\dag Z)]\, -\, \text{Tr}\,[\,\rho_0\, Z^\dag Z]\, =\\
=\, \text{Tr}\,[\zeta^\dag(\rho_0)\, Z^\dag Z]\, -\, \text{Tr}\,[\,\rho_0\, Z^\dag Z]\, =\\
=\, \text{Tr}\,[\,\rho_0\, Z^\dag Z]\, -\, \text{Tr}[\,\rho_0\, Z^\dag Z]\, =\, 0\, . \label{Lind eq2} \end{multline}
Thanks to~\eqref{Lind eq1} and to the fact that $\rho_0>0$, this ensures that $\zeta(Z^\dag Z) - Z^\dag Z = 0$, that is $Z^\dag Z\in\eta_\zeta$.
\end{proof}

Theorem~\ref{Lind} shows that under our assumptions $\eta_\zeta$ is a linear complex subspace of matrices which is in addition closed for hermitian adjunction and matrix product. Such a set is a particular instance of what is called in mathematics a \emph{von Neumann algebra}. Now that we have proved that $\eta_\zeta$ is equipped with such a peculiar structure, we can exploit the powerful characterization theorems holding for these algebras. Let us recall the main result, that is the classification of all the finite--dimensional von Neumann algebras up to unitary isomorphisms. In what follows we shall adopt the shorthand $\mathcal{M}_n \equiv \mathcal{M}(n;\mathds{C})$ in order to develop a more compact notation.

\begin{thm}[Classification of Finite--Dimensional von Neumann Algebras] \label{vNa} $\\$
Let $\mathcal{A}$ be a finite--dimensional von Neumann algebra composed of bounded operators over a Hilbert space $\mathcal{K}$. Then there exist integers $n_1,\ldots,n_m\geq 1$, Hilbert spaces $\mathcal{K}_1,\ldots, \mathcal{K}_m$ and a unitary isomorphism
\begin{equation*} U\,:\ \mathcal{K}\ \longrightarrow\ \bigoplus_{i=1}^m\, \mathds{C}^{n_i}\otimes \mathcal{K}_i \end{equation*}
such that 
\begin{equation} U\,\mathcal{A}\ U^\dag\ =\ \bigoplus_{i=1}^m\ \mathcal{M}_{n_i} \otimes\, \mathds{1}_{\mathcal{K}_i}\ . \label{vNa iso} \end{equation}
\end{thm}

What does Theorem~\ref{vNa} mean in practice, when we deal with a von Neumann algebra formed by square matrices? In that case the thesis states that there exists an orthonormal basis ($U$) of the whole Hilbert space ($\mathcal{K}$) such that when written in \emph{that} basis ($U(\cdot)U^\dag$) \emph{all} the matrices of our von Neumann algebra ($\mathcal{A}$) are at the same time cast into a block--diagonal form ($\bigoplus$), where each block corresponds to a subspace to which a structure of tensor product ($\mathds{C}^{n_i}\otimes \mathcal{K}_i$) can be given in such a way as to ensure that the matrices of our algebra are exactly the operators ($\mathcal{M}_{n_i}$) acting nontrivially only on the first space ($\mathds{C}^{n_i}$).

Theorem~\ref{vNa} characterizes the structure of the fixed subspace of an appropriate class of completely positive maps. However, it is not yet clear what is the action of the map on a generic matrix (which is not a fixed point). The following theorem, which constitutes the final result of this section, answers this question.

\begin{thm} \label{fix unital} $\\$
Let $\zeta\in\mathbf{CPu}_d$ be a completely positive, unital map such that its hermitian adjoint $\zeta^\dag\in\mathbf{CPt}_d$ has a strictly positive definite fixed point. Then in an appropriate orthonormal basis the fixed subspace $\eta_\zeta$ takes the block form
\begin{equation} \eta_\zeta\ =\ \bigoplus_i\ \mathcal{M}_{d_i^{(1)}} \otimes\, \mathds{1}_{d_i^{(2)}}\ , \label{eta zeta str} \end{equation}
where the $d_i^{(1)},d_i^{(2)}$ are positive integers. Moreover, if $\zeta$ is in addition idempotent (that is, $\zeta^2=\zeta$), its action on a generic $X\in\mathcal{M}(d;\mathds{C})$ can be written as
\begin{equation} \zeta(X)\ =\ \bigoplus_i\ \text{\emph{Tr}}_{i,2}\ [\ P_i X P_i\, (\mathds{1}_{d_i^{(1)}}\otimes \rho_{i,2}) \ ]\ \otimes\ \mathds{1}_{d_i^{(2)}}\ , \label{zeta act} \end{equation}
where
\begin{itemize}
\item $P_i$ is the orthogonal projector onto the $i$th subspace, in accordance with the decomposition~\eqref{eta zeta str} (onto $U^\dag\ \mathds{C}^{n_i}\otimes \mathcal{K}_i$ in the language of Theorem~\ref{vNa});

\item $\rho_{i,2}$ is a $d_i^{(2)} \times d_i^{(2)}$ density matrix (in the language of Theorem~\ref{vNa}, it acts on $\mathcal{K}_i$);

\item the symbol $\text{\emph{Tr}}_{i,2}$ stands for the partial trace over the second factor of the restricted $i$th subspace as indicated in~\eqref{eta zeta str} (i.e. $\mathcal{K}_i$ in the language of Theorem~\ref{vNa}).
\end{itemize}

\end{thm}

\begin{proof}
The first claim~\eqref{eta zeta str} is a restatement of Theorem~\ref{vNa}, and so there is nothing new to prove. Since the second claim is by far less obvious, let us proceed step--by--step.

\begin{itemize}

\item \emph{Step 1: $\zeta$ preserves the blocks.} We are claiming that if the input matrix $X$ has nonzero elements only in the $i$th block, then the same happens to the output matrix $\zeta(X)$. Let us prove this statement as follows. If $P_i$ is the orthogonal projector onto the $i$th block, equation~\eqref{eta zeta str} claims that it must be a fixed point of $\zeta$, i.e. $\zeta(P_i)=P_i$. Denote by $\{M_k\}_k$ a collection of Kraus operators for $\zeta$, that is
\begin{equation*} \zeta(\cdot)\, =\, \sum_k\, M_k(\cdot)M_k^\dag\ . \end{equation*}
Since $M_k P_i M_k^\dag\geq 0$ for each $k$, if $\zeta(P_i)=P_i$ then it must be true that each $M_k$ maps the $i$th block into itself. This in turn implies that $\zeta$ preserves the block structure.

\item \emph{Step 2: every output of $\zeta$ belongs to $\eta_\zeta$}. This descends directly from the further assumption that $\zeta^2=\zeta$, which implies that for all $X\in\mathcal{M}_d$ we have
\begin{equation} \zeta\left(\zeta(X)\right)=\zeta^2(X)=\zeta(X)\ . \label{fix unital eq0} \end{equation}
Adopting the shorthand notation
\begin{equation} X_{i,1}\otimes Y_{i,2}\ =\ 0 \oplus\ldots \oplus 0 \oplus \underbrace{(X_{i,1}\otimes Y_{i,2})}_{i\text{th block}} \oplus\, 0 \oplus\ldots \oplus 0\, , \label{shorthand} \end{equation}
what we have proved since now (i.e. step 1 and~\eqref{fix unital eq0}) ensures that
\begin{equation} \zeta\,(X_{i,1}\otimes Y_{i,2})\ =\ \left(\,F_i (X_{i,1}, Y_{i,2})\,\right)_{i,1} \otimes \mathds{1}_{i,2}\ , \label{fix unital eq1} \end{equation}
where each
\begin{equation*} F_i:\ \mathcal{M}_{d_i^{(1)}} \times \mathcal{M}_{d_i^{(2)}} \longrightarrow \mathcal{M}_{d_i^{(1)}} \end{equation*}
is a bilinear function.

\item \emph{Step 3: the $F_i$s defined through~\eqref{fix unital eq1} act as}
\begin{equation} F_i(X,Y)\ =\ X\ \text{Tr}[\rho_{i,2}\, Y]\ , \label{fix unital eq2} \end{equation}
\emph{for some $d_i^{(2)} \times d_i^{(2)}$ density matrices $\rho_{i,2}$s}. Let $A\in\mathcal{M}_{d_i^{(2)}}$ be such that $0\leq A\leq \mathds{1}_{d_i^{(2)}}$. Then for all pure states $\ket{\alpha}\in\mathds{C}^{d_i^{(1)}}$ one has
\begin{multline} 0\ \leq\ F_i\,(\Ket{\alpha}\!\!\Bra{\alpha},\, A) \otimes \mathds{1}_{d_i^{(2)}}\ =\ \zeta\, \left( \Ket{\alpha}\!\!\Bra{\alpha} \otimes A \right)\ \leq\\
\leq\ \zeta\, \left( \Ket{\alpha}\!\!\Bra{\alpha} \otimes \mathds{1}_{d_i^{(2)}} \right)\ =\ \Ket{\alpha}\!\!\Bra{\alpha} \otimes \mathds{1}_{d_i^{(2)}}\ , \label{fix unital eq3} \end{multline}
where the last passage is a consequence of the structure of the set of fixed points. From~\eqref{fix unital eq3} we can see that it must be
\begin{equation*} 0\ \leq\ F_i\,(\Ket{\alpha}\!\!\Bra{\alpha},\, A)\ \leq\ \Ket{\alpha}\!\!\Bra{\alpha}\ , \end{equation*}
which in turn implies that \mbox{$F_i\,(\Ket{\alpha}\!\!\Bra{\alpha},\, A)$} is proportional to \mbox{$\Ket{\alpha}\!\!\Bra{\alpha}$}, that is
\begin{equation} F_i\,(\Ket{\alpha}\!\!\Bra{\alpha},\, A)\ =\ \Ket{\alpha}\!\!\Bra{\alpha}\, f_i (A) \label{fix unital eq4} \end{equation}
for some \emph{positive, linear, unital} functional \mbox{$f_i:\, \mathcal{M}_{d_i^{(2)}}\rightarrow \mathds{R}$}. It is well--known that such a functional must have the form $f_i(Y)=\text{Tr}[\rho_{i,2}\, Y]$ for some $d_i^{(2)} \times d_i^{(2)}$ density matrix $\rho_{i,2}$. Using this fact and taking linear combinations of~\eqref{fix unital eq4} yields~\eqref{fix unital eq2}.

\item \emph{Step 4: conclusion}. Putting all together,~\eqref{fix unital eq1} and~\eqref{fix unital eq2} give
\begin{multline} \zeta\,(X_{i,1}\otimes Y_{i,2})\ =\ \left(\,X\ \text{Tr}[\rho_{i,2}\, Y]\, \right)_{i,1} \otimes \mathds{1}_{i,2}\ =\\
=\ \left( \text{Tr}_2 \, [\, (X \otimes Y)\, (\mathds{1}_{d_i^{(1)}}\otimes \rho_{i,2}) \, ] \right)_{i,1} \otimes \mathds{1}_{i,2} \label{fix unital eq5}\ . \end{multline} 
Taking linear combinations, we can see that for all $Z\in\mathcal{M}_{d_i^{(1)}}\otimes \mathcal{M}_{d_i^{(2)}}$ we must have
\begin{equation} \zeta\,(Z_i)\ =\ \left( \text{Tr}_2 \, [\, Z\, (\mathds{1}_{d_i^{(1)}}\otimes \rho_{i,2}) \, ] \right)_{i,1} \otimes \mathds{1}_{i,2} \label{fix unital eq6}\ , \end{equation} 
where of course
\begin{equation*} Z_i\ \equiv\ 0 \oplus\ldots \oplus 0 \oplus \underbrace{Z}_{i\text{th block}} \oplus\, 0 \oplus\ldots \oplus 0\, . \end{equation*}
Note that~\eqref{fix unital eq6} specify the action of $\zeta$ on each block. Taking into account also step 1, we easily see that the global action of $\zeta$ is given exactly by~\eqref{zeta act}.

\end{itemize}

\end{proof}

\subsection{General theory for all quantum channels} \label{gen theory}

At the end of Subsection~\ref{spectrum} and in Subsection~\ref{psibar}, we showed that particular simplifying assumptions can be made, without loss of generality, in order to study the structure of the phase subspace of a completely positive, trace--preserving map. Instead, through Subsection~\ref{part theory} we developed the theory for the case in which these assumptions are added as hypotheses. The task we must accomplish now is to follow this path \emph{backward}, generalizing Theorem~\ref{fix unital} to general, completely positive, trace--preserving channels. We summarize the conclusive theory in the following theorem, which is substantially Theorem 6.16 of~\cite{WolfQC} or Theorem 8 of~\cite{Wolfinverse}.

\begin{thm}[Structure Theorem for the Phase Subspace of Quantum Channels] \label{phase sub} $\\$
Let $\phi\in\mathbf{CPt}_d$ a quantum channel. Then there exist a subspace $\mathcal{K}\subseteq\mathds{C}^d$ for which an orthonormal basis can be found, such that with respect to the decomposition \mbox{$\mathds{C}^d=\mathcal{K}\oplus\mathcal{K}^\perp$} the phase subspace $\chi_\phi$ has the block structure form
\begin{equation} \chi_\phi\ =\ \underbrace{\bigoplus_i\, \mathcal{M}_{d_i^{(1)}}\otimes \rho_{i,2}}_{\text{acting on } \mathcal{K}}\ \oplus\ \underbrace{0}_{\text{acting on } \mathcal{K}^\perp}\ , \label{chi} \end{equation}
where the $\rho_{i,2}$ are density matrices.

Moreover, for all the operators $X\oplus 0$ (accordingly to the decomposition $\mathcal{K}\oplus\mathcal{K}^\perp$), the action of the spectral projector $E_\phi$ (as defined in~\eqref{E phi}) is
\begin{equation} E_\phi (X\oplus 0)\ =\ \bigoplus_i\ \text{\emph{Tr}}_{i,2}\, [\, P_i X P_i \, ] \otimes \rho_{i,2}\ \oplus\ 0\ , \label{E phi act} \end{equation}
where $P_i$ is the orthogonal projector onto the $i$th subspace, in accordance with the decomposition~\eqref{chi}, and $\text{\emph{Tr}}_{i,2}$ stands for the partial trace over the second factor of the $i$th subspace.

Finally, let us specify the action of $\phi$ on its phase subspace $\chi_\phi$. Denote by
\begin{equation} X\ =\ \bigoplus_i\, \left(X_{i,1}\otimes \rho_{i,2}\right)\ \oplus\ 0 \label{chi eq1} \end{equation}
the generic operator $X\in\chi_\phi$, decomposed accordingly to~\eqref{chi}. There are $d_i^{(1)}\times d_i^{(1)}$ unitary matrices and a permutation $\pi$ over the set of indices $i$, exchanging only indices sharing the same dimension $d_i^{(1)}$, such that

\begin{equation} \phi(X)\ =\ \bigoplus_i\, \left(U_i\,X_{\pi(i),1}\, U_i^\dag \otimes \rho_{i,2}\right)\ \oplus\ 0 \label{chi eq2} \end{equation}

for all $X\in\chi_\phi$ written in the form~\eqref{chi eq1}.

\end{thm}

\begin{proof} $\\$
First of all, define the spectral projection $E_\phi$ as in~\eqref{E phi}, and remind (end of Subsection~\ref{spectrum}) that the phase subspace of $\phi$ coincides with the fixed subspace of $E_\phi$, that is $\chi_\phi=\eta_{E_\phi}$. Moreover, $E_\phi$ is idempotent (i.e. $E_\phi^2=E_\phi$).

Next, define the subspace $\mathcal{K}$ for $E_\phi$ exactly as in~\eqref{eq K}, and construct the quantum channel $\tilde{E}_\phi$ associated to $E_\phi$ as in~\eqref{psi bar}, that is
\begin{equation} X=x\oplus 0\quad \Rightarrow\ E_\phi(X)= \tilde{E}_\phi(x) \oplus 0\, , \label{E phi bar}\end{equation}
where all the block decompositions are understood to be in accordance to the space decomposition \mbox{$\mathds{C}^d=\mathcal{K}\oplus \mathcal{K}^\perp$}.
The discussion in Subsection~\ref{psibar} shows that $\tilde{E}_\phi$ has a strictly positive fixed point. Moreover, an easy consequence of~\eqref{E phi bar} is that also $\tilde{E}_\phi$, just like $E_\phi$, is idempotent. Another consequence, as already observed, is that \mbox{$\eta_{E\phi}=\eta_{\tilde{E}_\phi}\oplus 0$}.

Now we are in condition to apply the whole Theorem~\ref{fix unital} to $\tilde{E}_\phi^\dag$. Equation~\eqref{eta zeta str} gives us
\begin{equation} \eta_{\tilde{E}_\phi^\dag}\ =\ \bigoplus_i\ \mathcal{M}_{d_i^{(1)}} \otimes\, \mathds{1}_{d_i^{(2)}}\ . \label{eta E phi tilde} \end{equation}
Moreover, from~\eqref{zeta act} we immediately obtain that for all the operators $X$ acting on $\mathcal{K}$
\begin{equation} \tilde{E}_\phi^\dag(X)\ =\ \bigoplus_i\ \text{Tr}_{i,2}\ [\ P_i X P_i\, (\mathds{1}_{d_i^{(1)}}\otimes \rho_{i,2}) \ ]\ \otimes\ \mathds{1}_{d_i^{(2)}}\ . \label{E phi bar dag act} \end{equation}
By the very definition of the hermitian adjunction through \mbox{$\text{Tr}\,[A^\dag \psi(B)]\equiv\text{Tr}\,[\psi^\dag(A)^\dag B]$}, it is not difficult to see that
\begin{equation} \tilde{E}_\phi (X)\ =\ \bigoplus_i\ \text{Tr}_{i,2}\ [\, P_i X P_i\, ]\ \otimes\ \rho_{i,2}\ . \label{E phi bar act} \end{equation}
Equation~\eqref{E phi bar act} shows immediately that
\begin{equation*} \eta_{\tilde{E}_\phi}\ =\ \bigoplus_i \, \mathcal{M}_{d_i^{(1)}}\otimes \rho_{i,2}\, , \end{equation*}
yielding~\eqref{chi}. Observe that in general \mbox{$\eta_{\tilde{E}_\phi}\neq \eta_{\tilde{E}_\phi^\dag}$}, as anticipated. However, these two matrix subspaces have the same dimension (see~\eqref{adjoint fix}). Putting together~\eqref{E phi bar} and~\eqref{E phi bar act} we obtain also~\eqref{E phi act}.

Naturally, $\phi$ maps linearly $\chi_\phi$ into itself. In order to explicitly write the action of $\phi$ on its phase subspace, a crucial observation is that there exists a legitimate quantum channel which \emph{inverts} this action. This channel is nothing but the $I_\phi$ of~\eqref{I phi}, which indeed verifies \mbox{$\phi\, I_\phi = I_\phi\, \phi = E_\phi$}. Obviously, also $I_\phi$ maps linearly $\chi_\phi$ into itself. Let us examine the consequences of this observation. Take $\Ket{\alpha}\in\mathds{C}^{d_i^{(1)}}$ for some $i$, and consider the operator $A=\Ket{\alpha}\!\!\Bra{\alpha}_{i,1}\otimes \rho_{i,2}\in\chi_\phi$, which acts nontrivially only on the $i$th block (see~\eqref{shorthand}). Since it admits no convex decomposition in $\chi_\phi$, its image under $\phi$ must share this same property. Otherwise,
\begin{multline*} \phi(A) = \sum_j p_j B_j\, ,\quad B_j\in\chi_\phi\quad \Rightarrow\\
\Rightarrow\quad A=E_\phi(A)= I_\phi\,\phi\,(A)\, =\, \sum_j p_j\, I_\phi(B_j) \end{multline*}
is a nontrivial convex combination of $X$ in $\chi_\phi$, absurd. This shows that $\phi(A)$ must be contained inside a single block and must be of the form \mbox{$\phi(A)=\Ket{\beta}\!\!\Bra{\beta}_{j,1}\otimes \rho_{j,2}$} for some $j$. Note that the continuity of $\phi$ requires that $j$ depends only on $i$ and not on $\Ket{\alpha}$ (no ``jumps'' between different block are allowed). Moreover, since $\phi$ has to be linear and bijective when mapping $\chi_\phi$ into itself, $d_i^{(1)}=d_j^{(1)}$. This ensures that there exist a permutation $\pi$ exchanging only blocks with the same $d_i^{(1)}$ and quantum channels $\phi_{i,1}\in\mathbf{CPt}_{d_i^{(1)}}$ mapping pure states into pure states in a bijective way such that
\begin{equation*} \phi(X)=\bigoplus_i\, \left(\phi_{i,1}(X_{\pi(i),1}) \otimes \rho_{i,2}\right)\ \oplus\ 0 \end{equation*}
for all $X$ written in the form~\eqref{chi eq1}. But it is well--known (Theorem~\ref{unitary}) that the only channels $\phi_{i,1}\in\mathbf{CPt}_{d_i^{(1)}}$ mapping pure states into pure states in a bijective way are the unitary evolutions. The same conclusion can be drawn if we note that the invertibility of $\phi$ in $\chi_\phi$ by means of a quantum channel ($I_\phi$) implies that the $\phi_{i,1}$s must have completely positive inverse., and we apply again Theorem~\ref{unitary}. Anyway, the unitarity of the $\phi_{i,1}$s gives the thesis~\eqref{chi eq2}.

\end{proof}

From this very general theorem several consequences can be deduced. One of them, for instance, is the solution of the so--called \emph{inverse eigenvalue problem} for the peripheral spectrum, as given by Wolf et al. in~\cite{Wolfinverse}. In that paper, the analogous of Theorem~\ref{phase sub} is used to obtain a complete classification of all the peripheral spectra of completely positive, trace--preserving maps, as expressed as follows.

\begin{thm}[Peripheral Spectra of Quantum Channels] \label{PS} $\\$
Let $\phi\in\mathbf{CPt}_d$ be a quantum channel. Then there are integers $n_c,d_c\in\mathds{N}$ (labeled by an index $c\in C$) satisfying $\sum_c n_c d_c\leq d $, and vectors $\omega_c\in\mathds{C}^{d_c}$ whose component are phases (i.e. $|\omega_{c \alpha}|\equiv 1\ \forall\ c,\alpha$), such that the peripheral spectrum of $\phi$ is
\begin{multline} \sigma_P(\phi)\, =\, \{\, \omega_{c \alpha}\, \omega_{c \beta}^*\ e^{\frac{2\pi i m_c}{n_c}}:\ c\in C,\\
0\leq m_c\leq n_c -1,\ 1\leq \alpha,\beta\leq d_c\, \}\, . \label{ps Wolf} \end{multline}
In this way the total number $|\sigma_P(\phi)|$ of peripheral eigenvalues of $\phi$ (counting multiplicities) is
\begin{equation} |\sigma_P(\phi)|\, =\, \sum_c n_c d_c^2\, . \end{equation}
Conversely, every set of numbers as in~\eqref{ps Wolf} is the peripheral spectrum of some \mbox{$\phi\in\mathbf{CPt}_{\sum_c n_c d_c}$}, which in addition can be chosen unital and with no other nonzero eigenvalue.
\end{thm}

\begin{proof} $\\$
Once we have the explicit form of the action of $\phi$ on $\chi_\phi$, as expressed in~\eqref{chi eq2}, it is not too difficult to see how~\eqref{ps Wolf} descends. 
The main ingredients to be used are the following.

\begin{itemize}

\item The spectrum of a unitary channel \mbox{$\mathcal{U}(\cdot)=U(\cdot)U^\dag$} acting on a $d$--dimensional system is
\begin{multline} \sigma(\mathcal{U})\, =\, \{ \omega_\alpha\, \omega_\beta^*:\ 1\leq \alpha,\beta\leq d\, \}\, =\\
=\, \sigma(U\otimes U^*)\, =\, \sigma(U) \times \sigma(U)^*\, , \label{ps Wolf eq1} \end{multline}
where $\sigma(U)=\{\omega_\alpha :\ 1\leq\alpha\leq d \}$ is the spectrum of $U$ (composed by phases). This can be explicitly seen by letting $\phi$ act on the operators $\Ket{\omega_\alpha}\!\!\Bra{\omega_\beta}$, where $U\Ket{\omega_\alpha}=\omega_\alpha \Ket{\omega_\alpha}$.

\item Every permutation can be decomposed into a product of cycles acting on \emph{disjoint} input subsets. For instance, $(1,2,3,4,5)\rightarrow (4,5,1,3,2)$ is the simultaneous action of the $3$--cycle $(1,3,4)\rightarrow (4,1,3)$ and of the $2$--cycle $(2,5)\rightarrow (5,2)$. The matrix associated with the permutation can be unitarily brought in block diagonal form, with each block corresponding to a cycle. Therefore, the global spectrum is just the union of the spectra of the different cycles.


\item The spectrum a $n$--cyclic block matrix of the form

\begin{equation} A\ =\ \begin{pmatrix} 0 & 0 & \ldots & 0 & A_n \\ A_1 & 0 & & & 0 \\ 0 & A_2 & \ddots & & \vdots \\ \vdots & \vdots & & \ddots & \vdots \\ 0 & 0 & \ldots & \ldots & 0 \end{pmatrix}\ , \label{ps Wolf eq2} \end{equation}

where $A_i\in\mathcal{M}(d;\mathds{C})$, is composed of all the $n$th complex roots of all the eigenvalues of $A_n\ldots A_2 A_1$. This can be verified (supposing by continuity $A_n\ldots A_1$ diagonalizable) either by looking for eigenvectors of the form \mbox{$(\lambda^{n-1}\,x^T,\,	\lambda^{n-2}\, x^T A_1,\,\ldots,\,x^T A_1^T \ldots A_{n-1}^T)^T$}, where $x$ is an eigenvector of $A_n\ldots A_1$ with eigenvalue $\lambda^n$, or by proving by induction the determinant formula 
\begin{equation*} \det (A-x\mathds{1}) = (-1)^{(n-1)d}\, \det(A_n\ldots A_1-x^n\mathds{1})\, . \end{equation*}
\end{itemize}

Now, $\phi$ acts on $\chi_\phi$ as a direct sum of $n_c$--cyclic $d_c^2 n_c\times d_c^2 n_c$ block matrices as~\eqref{ps Wolf eq2}, where each $A_i$ is a $d_c^2 \times d_c^2$ matrix representing a unitary evolution (these $d_c$s are the $d_i^{(1)}$s of Theorem~\ref{phase sub}, possibly repeated). As a consequence, also $A_n\ldots A_1$ is a unitary evolution, and we know that its spectrum is given by \mbox{$\{ \nu_\alpha\, \nu_\beta^*:\ 1\leq \alpha,\beta\leq d_c\, \}$}. Choosing phases $\omega_\alpha$ such that $\omega_\alpha^{n_c}=\nu_\alpha$ gives us exactly equation~\eqref{ps Wolf}. In order to construct examples of channels having a required peripheral spectrum of the form~\eqref{ps Wolf}, one has only to run this reasoning backward, using projective measurements to avoid any other nonzero eigenvalue.

\end{proof}

In what follows, we do not need the whole power of Theorem~\ref{PS}. Instead, we will find very useful a simple, nice consequence of it. \\

\begin{cor} \label{cor Wolf} $\\$
Let $\phi\in\mathbf{CPt}_d$ be a quantum channel satisfying \mbox{$|\sigma_P(\phi)|\geq 2$} (where $|\sigma_P(\phi)|$ is the number of peripheral eigenvalues of $\phi$, counting multiplicities). Then there exists an integer $1\leq n\leq d$ such that $1$ belongs to $\sigma_P(\phi^n)$ with multiplicity strictly greater than $1$.
\end{cor}

\begin{proof}
Let us assume that the multiplicity of $1\in\sigma_P(\phi)$ is exactly $1$ (otherwise it will be sufficient to choose $n=1$). Since $1$ is reached in \eqref{ps Wolf} for each $\alpha=\beta,\, m_c=0$, there must be only one possible $c$ (call it $0$), and moreover $d_0=1$. But then there exists $2\leq n_0 \leq d$ such that
\begin{equation*} \sigma_P(\phi)\, =\, \{\, e^{\frac{2\pi i m_0}{n_0}}\, :\ 0\leq m_0 \leq n_0 - 1\, \}\, . \end{equation*}
As a consequence, $1$ belongs to $\sigma_P (\phi^{n_0})$ with multiplicity $n_0\geq 2$.
\end{proof}

Another useful fact to be taken in mind is the statement of Theorem~\ref{PS} for qubit channels:
\begin{multline} \phi\in\mathbf{CPt}_2\ \ \Rightarrow\ \ \sigma_P(\phi)\ =\ \{1\},\ \{1,1\},\\
\{1,-1\},\ \{1,1,e^{i\theta},e^{-i\theta}\}\, . \label{Wolf q} \end{multline}
The commas in the preceding equation identify the possible alternative spectra. Recalling also~\eqref{unitary}, we can see that the last spectrum is the signature of a unitary evolution:
\begin{equation} \sigma_P(\phi)\, =\, \{1,1,e^{i\theta},e^{-i\theta}\} \ \ \Leftrightarrow\ \ \phi\in\boldsymbol{\mathcal{U}}_2\, .  \label{unitary from spectrum} \end{equation}
This makes sense, because $\{1,e^{i\theta},e^{-i\theta}\}$ is exactly the spectrum of a rotation in $\text{SO}(3)$, and~\eqref{spectrum q} holds.

\section{Universal Entanglement--Preserving Channels} \label{sec UEP}

Entanglement--breaking channels represent the most detrimental form of noise a quantum system can undergo. Not surprisingly they have been extensively studied in the literature and a complete
characterization of their properties have been obtained -- see e.g. Sec.~\ref{sec eb}. 
The natural counter-part of these maps is constitute by those transformations which are always innocuous, in the sense that \emph{they never break the entanglement} between Alice and Bob when acting locally on Alice's subsystem, no matter how weak it could be (provided that it existed in the initial state). In this section we focus on these special transformations proving that they coincide with the set of unitary maps, a result
which one could have guessed on physical ground but, to be best of our knowledge was never formalized before.

\begin{Def} \label{UEP} $\\$
Let $\phi\in\mathbf{CPt}$ be a quantum channel acting on system $A$. We say that $\phi$ is \emph{universal entanglement--preserving} (UEP) if for each quantum system $B$ and for each global entangled state $\rho_{AB}$, $(\phi\otimes I)(\rho_{AB})$ is again entangled:
\begin{equation} \rho_{AB}\notin\mathcal{S}_{AB}\ \ \Rightarrow\ \ (\phi\otimes I)(\rho_{AB})\notin\mathcal{S}_{AB}\ . \label{UEP eq}\end{equation}
\end{Def}

The requirement that the entanglement preservation must hold \emph{for all the states of the system (even if mixed)} is crucial. As noted in~\cite{Vnonmaxent}, we can not restrict this property to the pure states alone. Indeed, this would modify Definition~\ref{UEP} in such a way as to include other channels. For example, in the case of qubit, the channels that preserve the entanglement of every pure state are all but the entanglement--breaking ones (as can be immediately seen by observing that every entangled pure state is obtained from a maximally entangled one by allowing Bob to use a local, invertible filter). Instead, we shall see that Definition~\ref{UEP} is by far more strict.

We remark how the concept of universal entanglement--preserving channel is in some sense \emph{complementary to that of entanglement--breaking channel}. As the latter \emph{always destroys} the entanglement, the former \emph{always preserves} it, no matter how much entangled the input state is. A class of trivial examples of UEP channels is composed by the unitary evolutions. The rest of this section is devoted to the proof that \emph{the unitary evolutions are the only universal entanglement--preserving channels}.

The first tool we need is the following technical lemma, concerning the boundary of the convex set $\mathcal{S}_{AB}$ of separable states on a bipartite quantum system $AB$. Recall that a point $p$ belonging to a set $S$ (in a normed vector space, for instance) is said to be \emph{internal} to $S$ if there exists a ball of arbitrary, nonzero radius centered in $p$ and entirely contained in $S$. The non--internal point of $S$ form the \emph{boundary} of $S$, indicated by $\partial S$. A simple, common way of proving that a point $q$ indeed belongs to the boundary of $S$ is to find a curve $q_\varepsilon$ laying outside $S$ for $\varepsilon>0$ but having $q$ as limit when $\varepsilon\rightarrow 0$.

\begin{prop} \label{internal separable det 0} $\\$
Let $\rho_A, \rho_B$ be densities matrices on systems $A,B$. Denote by $\partial\mathcal{S}_{AB}$ the boundary of the set of separable states on the bipartite system $AB$. Then
\begin{equation} \rho_A\otimes\rho_B \in \partial\mathcal{S}_{AB}\ \ \Leftrightarrow\ \ \det\rho_A\, \det\rho_B\, =\, 0\, . \label{internal} \end{equation}
\end{prop}

\begin{proof}
Suppose that $\det\rho_A \det\rho_B > 0$ (that is, both $\rho_A$ and $\rho_B$ are strictly positive definite). Then $\rho_A\otimes\rho_B$ can be written as a nontrivial convex combination of the maximally mixed state $\frac{\mathds{1}}{d_A d_B}$ and a separable (actually, factorized) state. On one hand it is known (see for example~\cite{GurvitsBarnum}) that the maximally mixed state is always internal to the separable set. On the other hand, a simple reasoning, valid for all the convex sets $S$, shows that a nontrivial convex combination of an internal point $p\in S-\partial S$ and another point $q\in S$ is again internal to $S$. The following geometric picture can help to visualize that reasoning. Since $p$ is internal to the convex set $S$, we can inscribe a circular cone inside $S$, whose axis contains $p$ and $q$ as endpoints, and our nontrivial convex combination as a non--extremal point. Now, every axial point different from the vertex and the base point must be internal to the cone, and so to $S$. This proves that $\rho_A\otimes\rho_B\in\partial\mathcal{S}_{AB}$ implies $\det\rho_A \det \rho_B = 0$.

Let us turn our attention to the converse statement. Suppose for instance that $\det \rho_A=0$; we must prove that $\rho_A\otimes\rho_B\in\partial\mathcal{S}_{AB}$. Take a vector $\Ket{1}$ such that $\rho_A\Ket{1}=0$, and $\Ket{2}\perp\Ket{1}$. Consider
\begin{equation*} \Ket{\Psi}\,\equiv\,\frac{\Ket{11}+\Ket{22}}{\sqrt{2}}\ ,\quad \rho_\varepsilon\,\equiv\, \varepsilon \Ket{\Psi}\!\!\Bra{\Psi}+(1-\varepsilon)\, \rho_A\otimes\rho_B\ , \end{equation*}
where $0<\varepsilon\leq 1$.
Then
\begin{itemize}
\item $\rho_\varepsilon$ is a density matrix for each $\varepsilon\geq 0$.
\item $\lim_{\varepsilon\rightarrow 0} \rho_\varepsilon = \rho_A\otimes\rho_B$ .
\item Acting with partial transposition $T_B$ on $\rho_\varepsilon$ produces a non--positive operator, for each $\varepsilon>0$; as a consequence, the PPT criterion of separability~\cite{PeresPPT,HorodeckiPPT} states that $\rho_\varepsilon$ can not be separable for $\varepsilon>0$. To see that $\rho_\varepsilon^{T_B}$ is not positive, we will prove that its restriction to the subspace $W\equiv \text{Span}\, \{ \Ket{12},\Ket{21} \}$ has negative determinant. In fact, simple calculations show that
\begin{equation*} \rho_\varepsilon^{T_B}|_W\ =\ \begin{pmatrix} 0 & \varepsilon/2 \\ \varepsilon/2 & (1-\varepsilon)\,a \end{pmatrix} \ ,\end{equation*}
where $a$ is a real number given by
\begin{equation*} a\,\equiv\, \Bra{2}\rho_A\Ket{2}\, \Bra{1}\rho_B\Ket{1}\, . \end{equation*}
As claimed, $\det (\rho_\varepsilon^{T_B}|_W) < 0$ for all $\varepsilon>0$, and so $\rho_\varepsilon^{T_B}$ can not be positive definite in that range.
\end{itemize}

Since we can construct entangled matrices arbitrary close to $\rho_A\otimes\rho_B$, it must be $\rho_A\otimes\rho_B\in\partial\mathcal{S}_{AB}$ .
\end{proof}

Before we arrive at the final result of this section, let us formalize a little, nice lemma which will turn out to be useful many times.

\begin{lemma} \label{ZZ} $\\$
Let $Z$ be a non--normal complex matrix (that is, \mbox{$[Z,Z^\dag]\neq 0$}) of size $d\times d$. Define 
\begin{equation} Q(Z)\, \equiv\, \begin{pmatrix} \mathds{1} & Z \\ Z^\dag & Z^\dag Z \end{pmatrix}\ . \label{ZZ eq} \end{equation}
Then $Q(Z)$ is an (unnormalized) entangled state over the Hilbert space $\mathds{C}^d\otimes \mathds{C}^2$.
\end{lemma}

\begin{proof}
Since
\begin{equation*} Q(Z)\, =\, \begin{pmatrix} \mathds{1} & Z \end{pmatrix}^\dag\, \begin{pmatrix} \mathds{1} & Z \end{pmatrix}\ , \end{equation*}
it is obvious that $Q(Z)\geq 0$ is a legitimate unnormalized state. To show that it is entangled, it suffices to apply the PPT criterion and argue that $Q(Z)^{T_B}$ is not positive. A straightforward calculation gives
\begin{equation*} \left(Q(Z)^{T_B}\right)^*\, =\, \begin{pmatrix} \mathds{1} & Z^\dag \\ Z & Z^\dag Z \end{pmatrix}\ . \end{equation*}
Now, it is well--known (see for instance~\cite{HJ1}, p. 472) that a block matrix $\left(\begin{smallmatrix} A & B \\ B^\dag & C \end{smallmatrix}\right)$ with $A>0$ is positive if and only if $C\geq B^\dag A^{-1} B$. In our case this would give the condition $Z^\dag Z \geq Z Z^\dag$, that is $[Z,Z^\dag]\leq 0$. Since $\text{Tr}\,[Z,Z^\dag]=0$, this would imply that $[Z,Z^\dag]=0$, which is forbidden by hypothesis.

\end{proof}

Now we are in position to state and prove the main result about the universal entanglement--preserving channels. This is the content of the following theorem. 

\begin{thm}[Classification of UEP Channels] \label{UEP U} $\\$
The only universal entanglement--preserving channels are the unitary evolutions.
\end{thm}

\begin{proof}
We already know that a unitary channel is definitely UEP. The problem is to prove the converse statement. So, let $\phi$ be a universal entanglement--preserving channel.

\begin{enumerate}

\item \emph{Step 1: $\phi(\mathds{1})>0$}. Suppose by contradiction that there exists a pure state $\Ket{\alpha}$ such that $\Bra{\alpha}\phi(\mathds{1})\Ket{\alpha}=0$. Then we can rewrite this condition as $\text{Tr}[\,\phi^\dag(\Ket{\alpha}\!\!\Bra{\alpha})\,]=0$, that is (using the fact that $\phi^\dag$ is a positive map, too) $\phi^\dag(\Ket{\alpha}\!\!\Bra{\alpha})=0$. This shows that $\phi^\dag$ admits a zero eigenvalue (see Subsection~\ref{spectrum}), which in turn implies that also $\phi$ admits a zero eigenvalue (because $\sigma(\phi)=\sigma(\phi^\dag)^*$). Then, consider an input matrix $X$ such that $\phi(X)=0$. Theorem~\ref{spect prop Pt} ensures that $X$ can be chosen hermitian and traceless, and this in particular shows that it can not be a multiple of an orthogonal projector. This last property allows us to choose another hermitian matrix $Y$ whose support is entirely contained in that of $X$, and such that $[X,Y]\neq 0$. Defining for $Z_\varepsilon\equiv X+i\varepsilon Y$ (so that $[Z_\varepsilon,Z_\varepsilon^\dag]\neq 0$ if $\varepsilon>0$) allows us to exploit Lemma~\ref{ZZ} to say that $Q(Z_\varepsilon)$ (as defined in~\eqref{ZZ eq}) is an (unnormalized) entangled state for all $\varepsilon>0$. Then
\begin{multline*} (I\otimes\phi)\left(Q(Z_\varepsilon)\right)\ =\ \begin{pmatrix} \phi(\mathds{1}) & i \varepsilon\, \phi(Y) \\ -i \varepsilon\, \phi(Y) & \phi\,(Z_\varepsilon^\dag Z_\varepsilon) \end{pmatrix}\ =\\
=\ (I\otimes\phi)\, \begin{pmatrix} \mathds{1} & i \varepsilon Y \\ -i \varepsilon Y & Z_\varepsilon^\dag Z_\varepsilon \end{pmatrix} \end{multline*}
must be again entangled for all $\varepsilon>0$, because of the very definition of UEP channel. However, we will show in a moment that it is instead \emph{separable} for sufficiently small $\varepsilon>0$. This goal will be reached by arguing that indeed already $\left(\begin{smallmatrix} \mathds{1} & i \varepsilon Y \\ -i \varepsilon Y & Z_\varepsilon^\dag Z_\varepsilon \end{smallmatrix}\right)$ is a positive, separable (unnormalized) state for sufficiently small $\varepsilon>0$. Since everything happens inside the support of $X$, we can simply project down here and prove the separability of the resulting matrix. This is the same as to suppose $\det X\neq 0$, i.e. $X^2>0$. With this hypothesis, we see that 
\begin{equation*} \lim_{\varepsilon\rightarrow 0}\ \begin{pmatrix} \mathds{1} & i \varepsilon Y \\ -i \varepsilon Y & Z_\varepsilon^\dag Z_\varepsilon \end{pmatrix}\ =\ \begin{pmatrix} \mathds{1} & 0 \\ 0 & X^2 \end{pmatrix} \end{equation*}
is indeed a nontrivial convex combination of the completely mixed state (which as already mentioned is \emph{internal} to the set of separable state, as proved in~\cite{GurvitsBarnum}) and another separable matrix. By virtue of the same reasoning we used in proving Proposition~\ref{internal separable det 0}, we can conclude that $\left(\begin{smallmatrix} \mathds{1} & 0 \\ 0 & X^2 \end{smallmatrix}\right)$ is again \emph{internal} to the separable set, and so that its approximation $\left(\begin{smallmatrix} \mathds{1} & i \varepsilon Y \\ -i \varepsilon Y & Z_\varepsilon^\dag Z_\varepsilon \end{smallmatrix}\right)$ is separable for sufficiently small $\varepsilon>0$.

\item \emph{Step 2: making $\phi$ unital}. Thanks to step 1, let us define the map $\psi$ by
\begin{equation} \psi(X)\,\equiv\ \phi(\mathds{1})^{-1/2}\,\phi(X)\,\phi(\mathds{1})^{-1/2}\, . \label{UEP eq psi} \end{equation}
Observe that:
\begin{itemize}
\item $\psi$ is again completely positive;
\item $\psi$ is \emph{unital} (even if no longer trace--preserving).
\item $\psi$ is again UEP (the definition of UEP makes sense also for non--trace--preserving maps), because conjugation by means of an invertible matrix does not affect separability.
\end{itemize}

\item \emph{Step 3: $\psi$ preserves non--invertibility of nonnegative matrices}. We claim that
\begin{equation} \forall\, \rho\geq 0\, ,\ \det \rho=0\ \Rightarrow\ \det\psi(\rho)=0\, . \label{UEP eq1} \end{equation}
We can suppose that $\rho$ is a normalized density matrix. Consider a second system $B$ of dimension $d_B\geq 2$. If $\rho\geq 0$ but $\det \rho =0$, we know from Proposition~\ref{internal separable det 0} that \mbox{$\rho\otimes \frac{\mathds{1}}{d_B}\in\partial\mathcal{S}_{AB}$}. This entails that one can construct a curve $R_\varepsilon$ ($0< \varepsilon\leq 1$) composed of \emph{entangled states} of $AB$ such that  
\begin{equation*} \lim_{\varepsilon\rightarrow 0^+}\, R_\varepsilon\ =\ \rho\otimes\frac{\mathds{1}}{d_B}\ . \end{equation*}
Since $\psi$ is UEP, it must be $(\psi\otimes I)(R_\varepsilon)\notin\mathcal{S}_{AB}$ for each $\varepsilon>0$. Moreover, observe that
\begin{multline*} \lim_{\varepsilon\rightarrow 0^+}\, (\psi\otimes I)(R_\varepsilon)\ =\ (\psi\otimes I)\,\left(\lim_{\varepsilon\rightarrow 0^+}\, R_\varepsilon\right)\ =\\
=\ (\psi\otimes I)\,\left(\rho\otimes\frac{\mathds{1}}{d_B}\right)\ =\ \psi(\rho)\otimes\frac{\mathds{1}}{d_B}\ \in\ \mathcal{S}_{AB}\ . \end{multline*}
Strictly speaking, $\psi(\rho)$ is no longer a density matrix, because it is not guaranteed to have unit trace. Anyway, it makes sense to say that its normalized form is indeed separable. We have proved that there exists a curve composed of entangled states whose limit is the separable state $\psi(\rho)\otimes\frac{\mathds{1}}{d_B}$. This is the same as to say that $\psi(\rho)\otimes\frac{\mathds{1}}{d_B}\in\partial\mathcal{S}_{AB}$, and so Proposition~\ref{internal separable det 0} implies that
\begin{equation*} \det\, \psi(\rho)\ =\ 0\, . \end{equation*}

\item \emph{Step 4: $\psi$ preserves non--invertibility of general hermitian matrices}. We have to prove that
\begin{equation} \forall\, X=X^\dag\ ,\  \det X=0\ \Rightarrow\ \det\psi(X)=0\, . \label{UEP eq2}\end{equation}
Let us exploit step 3 by applying~\eqref{UEP eq1} to the positive matrix $X^2$:
\begin{multline} \det X=0\ \Rightarrow\ \det(X^2)=0\ \Rightarrow\ \det\psi(X^2)=0 \ \Rightarrow \\
\Rightarrow\ \exists \Ket{\eta}\, :\ \Braket{\eta | \psi(X^2) | \eta}=0\, . \label{UEP eq2b}\end{multline}
Since $\psi$ is positive and unital, we can apply the Kadison's inequality~\eqref{Kadison P} and argue that
\begin{equation} \Braket{\eta | \psi(X^2) | \eta}\ \geq\ \left(\, \psi(X)\Ket{\eta}\, \right)^\dag\ \left(\, \psi(X)\Ket{\eta}\, \right)\, . \label{UEP eq2c} \end{equation}
Naturally,~\eqref{UEP eq2b} and~\eqref{UEP eq2c} together show exactly that there exists $\Ket{\eta}$ such that $\psi(X)\Ket{\eta}=0$, i.e. that $\det\psi(X)=0$.

\item \emph{Step 5: $\psi$ preserves the spectrum}. From now on, we can proceed on the guideline drawn by~\cite{WolfQC} (see Chap. 3, p.~66). A crucial fact is that $\psi$ must preserve the spectrum of an hermitian matrix \emph{as a set}, i.e. that
\begin{equation} \forall\, X=X^\dag ,\ \lambda\in\sigma(X) \ \Rightarrow\ \lambda\in\sigma\left( \psi(X) \right)\, . \label{UEP eq3}\end{equation}
Indeed, by~\eqref{UEP eq2} one has
\begin{multline*} \lambda\in\sigma(X) \ \Rightarrow\ \det(X-\lambda\mathds{1})=0\ \Rightarrow\ \det\psi(X-\lambda\mathds{1})=0\ \Rightarrow\\
\Rightarrow\ \det\left( \psi(X)-\lambda\mathds{1} \right) = 0\ \Rightarrow\ \lambda\in\sigma\left( \psi(X) \right)\, . \label{UEP eq3}\end{multline*} 
Observe that~\eqref{UEP eq3} implies that also the \emph{multiplicities} of the eigenvalues are the same for $X$ and $\psi(X)$ (i.e. $\psi$ preserves the spectra as \emph{multisets}). Indeed, the set of hermitian matrices with non--degenerate spectrum is dense in $\mathcal{H}(d;\mathds{C})$ (as can be easily seen by perturbing the eigenvalues). Take a sequence $X_\varepsilon$ (with $0<\varepsilon\leq 1$) of hermitian matrices enjoying this property, and such that $\lim_{\varepsilon\rightarrow 0^+}=X$. Denote as usual by $\sigma(\cdot)$ the spectrum as a multiset (i.e. counting multiplicities). Then, from~\eqref{UEP eq3} we deduce that
\begin{equation*} \sigma(\psi(X_\varepsilon))\, \equiv\, \sigma(X_\varepsilon)\quad \forall\, \varepsilon>0\, . \end{equation*}
On the other hand, the continuity of the eigenvalues requires that
\begin{multline*} \sigma(\psi(X))\, =\, \sigma\left( \lim_{\varepsilon\rightarrow 0} \psi(X_\varepsilon) \right)\, =\, \lim_{\varepsilon\rightarrow 0} \sigma(\psi(X_\varepsilon))\, =\\ 
=\, \lim_{\varepsilon\rightarrow 0} \sigma(X_\varepsilon)\, =\, \sigma(\lim_{\varepsilon\rightarrow 0} X_\varepsilon)\, =\, \sigma(X)\, . \end{multline*}
In taking the above limits we have implicitly understood an obvious metric over the spectra.

\item \emph{Step 6: preparing the ground for Wigner's theorem}. We now claim that $\psi$ sends pure states into pure states in such a way as to preserve the moduli of the scalar products:
\begin{gather} \psi\, (\Ket{\alpha}\!\!\Bra{\alpha})\,\equiv\,\Ket{\alpha'}\!\!\Bra{\alpha'}\quad \forall \Ket{\alpha}\, , \label{UEP eq5a}\\
|\Braket{\alpha|\beta}|\,\equiv\,|\Braket{\alpha'|\beta'}|\quad \forall\, \Ket{\alpha},\Ket{\beta}\, . \label{UEP eq5b} \end{gather}
The proof is as follows. On one hand, for each $\Ket{\alpha}$ we have
\begin{multline*} \sigma\,(\Ket{\alpha}\!\!\Bra{\alpha})\, =\, \{1,\underbrace{0,\ldots,0}_{d^2-1}\}\ \Rightarrow\\
\Rightarrow\ \sigma\,(\,\psi(\Ket{\alpha}\!\!\Bra{\alpha})\,)\, =\, \{1,\underbrace{0,\ldots,0}_{d^2-1}\}\ \Rightarrow\\
\Rightarrow\ \psi(\Ket{\alpha}\!\!\Bra{\alpha}) = \Ket{\alpha'}\!\!\Bra{\alpha'}\, . \end{multline*}
On the other hand, take two states $\Ket{\alpha},\Ket{\beta}$, and denote by $\Ket{\alpha'},\Ket{\beta'}$ their images under the action of $\psi$. Then
\begin{multline*} \{\, 1+|\Braket{\alpha|\beta}|\, ,\ 1-|\Braket{\alpha|\beta}|\, ,\ \underbrace{0,\ldots,0}_{d^2-2}\, \}\ =\\
=\ \sigma\, (\,\Ket{\alpha}\!\!\Bra{\alpha}+\Ket{\beta}\!\!\Bra{\beta}\,)\ =\\
=\ \sigma\, \left(\,\psi(\,\Ket{\alpha}\!\!\Bra{\alpha}+\Ket{\beta}\!\!\Bra{\beta}\,)\,\right)\ =\\
=\ \sigma\, \left(\,\Ket{\alpha'}\!\!\Bra{\alpha'}+\Ket{\beta'}\!\!\Bra{\beta'}\,\right)\ =\\
=\ \{\, 1+|\Braket{\alpha'|\beta'}|\, ,\ 1-|\Braket{\alpha'|\beta'}|\, ,\ \underbrace{0,\ldots,0}_{d^2-2}\, \}\ ,
\end{multline*}
from which we deduce exactly
\begin{equation*} |\Braket{\alpha|\beta}|\, =\, |\Braket{\alpha'|\beta'}|\, . \end{equation*}

\item \emph{Step 7: conclusion}. Thanks to~\eqref{UEP eq5b}, the hypothesis of Wigner's Theorem~\ref{Wigner} are satisfied. Since an anti--unitary transformation can be represented as the complex conjugation in some basis followed by a unitary operation, we must conclude that for all vectors $\Ket{\alpha}$
\begin{equation*} \Ket{\alpha'}\,=\, e^{i\varphi(\alpha)}\,U\Ket{\alpha}\quad \text{or}\quad \Ket{\alpha'}\,=\, e^{i\varphi(\alpha)}\,U\Ket{\alpha^*}\, , \end{equation*}
where $U$ is unitary. Therefore, one has
\begin{equation*} \psi\,(\Ket{\alpha}\!\!\Bra{\alpha})\, \equiv\, U\, \Ket{\alpha}\!\!\Bra{\alpha}\, U^\dag\end{equation*}
or
\begin{equation*} \psi\,(\Ket{\alpha}\!\!\Bra{\alpha})\,\equiv\, U\, \Ket{\alpha^*}\!\!\Bra{\alpha^*}\, U^\dag\, \equiv\, U \Ket{\alpha}\!\!\Bra{\alpha}^T\, U^\dag\, . \end{equation*}
Actually, this implies that for all the input matrices $X$ one has
\begin{equation*} \psi(X)\,\equiv\, UXU^\dag\quad \text{or}\quad \psi(X)\,\equiv\, UX^T U^\dag\, . \end{equation*}
The second option has to be discarded, because $\psi$ is completely positive, and the matrix transposition is only positive. Going back to $\phi$ by means of~\eqref{UEP eq psi}, we obtain
\begin{equation*} \phi(X)\, \equiv\ \phi(\mathds{1})^{1/2}\,UXU^\dag\,\phi(\mathds{1})^{1/2}\, . \end{equation*}
Since $\phi$ has to be trace--preserving, we can easily see that it must be $U^\dag\phi(\mathds{1})U = \mathds{1}$, that is $\phi(\mathds{1})=\mathds{1}$. Hence, we deduce exactly $\phi(X) \equiv UXU^\dag$ for all the input matrices $X$, as we claimed.

\end{enumerate}
\end{proof}

This proof of Theorem~\ref{UEP U}, which is only one of the several different proofs, is technically quite complex. However, this can not distract our attention from its physical meaning. Concerning the entanglement between Alice and Bob, Theorem~\ref{UEP U} says a simple, intuitive thing.
\begin{center}
\emph{A truly noisy interaction of Alice's subsystem with an external environment definitely breaks some form of entanglement between Alice and Bob.}
\end{center}

From a conceptual point of view, the context of our investigations is remarkably clarified. This characterization theorem can be seen as exactly specular to Theorem~\ref{str thm EB}. The latter specifies an operational meaning (the Holevo form~\eqref{Holevo form}) for those channels which always break the quantum correlations. The former, instead, claims that only the unitary evolutions can definitely make the entanglement survive.

All that strengthens our belief that the deterioration to which the entanglement is subjected can be used to quantify the amount of noise introduced by a quantum channel (as proposed in~\cite{V,LV}). In this respect, we have proved that this kind of measure is \emph{faithful}: if no entanglement is wasted, there is no true noise acting on the system. So, the main purpose of the following section is to further investigate the classifications induced on the set of quantum channels by the entanglement preservation properties.

\section{Entanglement--Saving Channels} \label{sec ES}

In our previous paper~\cite{LV}, we began to investigate the properties of quantum channels by means of the effect they have on the entanglement of a bipartite state. We proposed two (inverse) measures of noise for a generic channel $\phi$, namely the \emph{direct $n$--index} (already defined in~\cite{V}) and the \emph{filtered $\mathcal{N}$--index}. The former is simply the minimum number of consecutive iterations of $\phi$ necessary to obtain an entanglement--breaking channel. Instead, the latter index describes the optimized scenario, in which the application of intermediate quantum channels (called \emph{filters}) is allowed in order to save as long as possible the entanglement. In this paper we will focus only on the direct $n$--index.

\subsection{Statement of the problem}

We begin our study of this functional by posing the following question: which kind of noise is so weak that it never separates completely a maximally entangled state, even if applied an arbitrary number of times? Within the language developed through~\cite{V,LV}, these \emph{entanglement--saving} channels are characterized by an infinite value of the direct $n$--index. So we can give the following definition.

\begin{Def}[Entanglement--Saving Channels] \label{ES} $\\$
A quantum channel $\phi\in\mathbf{CPt}$ is called \emph{entanglement--saving} (ES) if $n(\phi)=\infty$, i.e. if
\begin{equation*} \phi^n\notin\mathbf{EBt}\quad\forall\ n\in\mathds{N}\ . \end{equation*}
\end{Def}

The main goal of the rest of this section is to find an adequate characterization of these channels. The objective will be completely achieved almost everywhere (i.e. apart from a set of zero measure), and in arbitrary dimension. As a corollary, the problem in the case of qubit will be completely solved.

The first, elementary property of the ES set is its closure under unitary conjugation, that is
\begin{equation} \phi\ \text{entanglement--saving}\ \ \Leftrightarrow\ \ \mathcal{U}\phi\,\mathcal{U}^\dag\ ,\quad \forall\, \mathcal{U}\in\boldsymbol{\mathcal{U}}\ . \label{nUC} \end{equation}
Moreover, observe that if a channel $\phi\in\mathbf{CPt}$ is \emph{not} entanglement--saving, then the elements of the sequence $(\phi^n)_{n\in\mathds{N}}$ become eventually entanglement--breaking for \emph{sufficiently large $n$}, i.e. for all the $n$ greater than or equal to a certain threshold (which is by definition the direct index $n(\phi)$).  

As we shall see, a remarkable simplification in the theory of entanglement--saving channels occurs if we restrict our analysis to the set of quantum channels $\phi$ such that $\det(\phi)\neq 0$. Here the determinant of a linear map is the product of its eigenvalues, as usual (see Subsection~\ref{spectrum}). Although this assumption could seem rather arbitrary, we will see that it is indeed quite natural at least for a single qubit. In fact, in that case it causes no loss of generality. In order to take advantage of this restriction, we need some preliminary results concerning the entanglement--breaking channels.

\subsection{Preliminaries}
It is well--known (Theorem~\ref{spect prop Pt}) that every quantum channel has a positive fixed point. However, this density matrix can be or not be \emph{strictly} positive. Recall that we distinguish between the two alternatives by calling \emph{strictly positive} a matrix $A>0$, and \emph{semipositive} a matrix $A\geq 0$ having at least one zero eigenvalue (i.e. satisfying $\det A=0$). Whenever this distinction is not necessary, we say simply \emph{positive}. 

To appreciate the importance of the question and its link with the separability problem, recall Proposition~\ref{internal separable det 0}. We will find this result quite useful also in this context. Indeed, we proved that matrices of the form $\rho_0\otimes\frac{\mathds{1}}{d}$, with $\rho_0$ only semipositive, belong to the boundary of the separable set. Therefore, for an entanglement--breaking channel (whose images are \emph{always} separable) the presence of a semipositive fixed point must be a rather delicate situation. Our immediate purpose is to discuss the consequences of this possibility. Actually, we will explore a more general circumstance through the following theorem.

\begin{thm}[Image of Semipositive Matrices Through EB Channels] \label{EB det 0} $\\$
Let $\phi\in\mathbf{EB}_d$ be an entanglement--breaking channel. Suppose that there exists a semipositive matrix $A\geq 0$, with $\text{\emph{rk}}\,A=r<d$, such that $\text{\emph{rk}}\, \phi(A)=s$ verifies $r^2+s^2<2dr$. Then
\begin{equation} \dim \ker \phi\, \geq\, 2dr-r^2-s^2\, >\, 0\, . \end{equation}
\end{thm}

\begin{proof}
Let us write the action of the entanglement--breaking channel $\phi$ in Holevo form~\ref{Holevo form}:
\begin{equation*} \phi(X)\ =\ \sum_{i\in I} \rho_i\ \text{Tr}\,XE_i\ . \end{equation*}
Here the $\rho_i$ are density matrices, and the operators $E_i$ are positive definite. Calling $A'\equiv\phi(A)$, one has
\begin{equation*} A'\ =\ \sum_{i\in I} \rho_i\ \text{Tr}\,A E_i\ . \end{equation*}
Recall that $\text{supp}\, X$ is the support of a hermitian matrix $X$, i.e. the orthogonal complement of the kernel. Clearly, thanks to the positivity of the operators $\rho_i$ and $E_i$,
\begin{equation} \forall\, i\in I\, ,\quad \text{supp}\, \rho_i\subseteq \text{supp}\, A'\quad \text{or}\quad \text{supp}\, E_i\subseteq \ker A\, . \label{EB condition} \end{equation}
Let us define a bipartition of $I$ distinguishing between these two possibilities:
\begin{equation*} I_0\,\equiv\, \{\, i\in I\, :\ \text{supp}\, \rho_i\subseteq \text{supp}\, A'\, \}\, ,\quad I_1\,\equiv\, I \backslash I_0\, . \end{equation*}
Thanks to \eqref{EB condition}, we can write
\begin{equation} \forall\, i\in I_1\ ,\quad \text{supp}\, E_i\subseteq \ker A\, . \label{EB conseq} \end{equation}
Denote by $\Ket{1},\ldots,\Ket{r}$ an orthonormal basis of $\text{supp}\,A$, and by $\Ket{1},\ldots,\Ket{d}$ one of its completions to a global orthonormal basis. Consider the vector subspace of hermitian matrices $V$ spanned on $\mathds{R}$ by the operators
\begin{itemize}
\item $\Ket{\alpha}\!\!\Bra{\beta}+\Ket{\beta}\!\!\Bra{\alpha}$, with $1\leq\min\{\alpha,\beta\}\leq r$ and $\alpha\neq \beta$;
\item $i \Ket{\alpha}\!\!\Bra{\beta}- i \Ket{\beta}\!\!\Bra{\alpha}$, again with $1\leq\min\{\alpha,\beta\}\leq r$ and $\alpha\neq \beta$;
\item $\Ket{\alpha}\!\!\Bra{\alpha}$, with $1\leq \alpha\leq r$.
\end{itemize}
Moreover, define also
\begin{equation*} W\,\equiv\,\{\, X=X^\dag\, :\ \ \text{supp}\, X\subseteq \text{supp}\, A'\, \}\, .\end{equation*}
The dimensions of $V$ and $W$ are easy to calculate:
\begin{gather*} \dim V\, =\, 2\, \sum_{k=1}^d (d-\alpha)\, +\, r\, =\, 2dr-r^2\ , \\
\dim W\, =\, ( \text{rk}\, A' )^2\, =\, s^2\ . \end{gather*}
Observe that
\begin{equation} \dim V - \dim W\, =\, 2dr - r^2 -s^2\, >\, 0 \label{disug dim} \end{equation}
by hypothesis. Moreover, \eqref{EB conseq} implies that
\begin{equation*} \forall\ i\in I_1,\ \ \forall\ X\in V,\quad \text{Tr}\, X E_i=0\, . \end{equation*}
Therefore, for all $X\in V$ we have
\begin{multline*} \phi(X)\ =\ \sum_{i\in I}\, \rho_i\, \text{Tr}\, X E_i\ =\\
=\ \sum_{i\in I_0}\, \rho_i\, \text{Tr}\, X E_i\ +\ \sum_{i\in I_1}\, \rho_i\, \text{Tr}\, X E_i\ =\\
=\ \sum_{i\in I_0}\, \rho_i\, \text{Tr}\, X E_i\ \in\ W\, . \end{multline*}
As a consequence, it makes sense to consider the restriction $\phi|_V: V\rightarrow W$. Thanks to \eqref{disug dim}, one has
\begin{multline*} \dim \ker \phi\, \geq\, \dim \ker \phi|_V\, =\, \dim V -\, \text{rk}\, \phi|_V\, \geq\\
\geq\, \dim V -\, \dim W\, =\, 2dr - r^2 -s^2\, >\, 0\, . \end{multline*}
\end{proof}

With Theorem~\ref{EB det 0}, we have explored a nontrivial feature of the entanglement--breaking channels. However, our result does not look conceptually transparent at all. In order to clarify its meaning, let us examine the following simpler corollary (which is what we really need in the rest of the paper). 

\begin{cor} \label{EB det 0 cor} $\\$
Let $\phi\in\mathbf{EBt}$ be an entanglement--breaking channel. Suppose that $\phi$ has a semipositive fixed point. Then
\begin{equation*} \dim \ker \phi \geq 2(d-1)\, , \end{equation*}
and in particular $\det\phi=0$.
\end{cor}

\begin{proof}
It suffices to apply Theorem~\ref{EB det 0} with $r=s\leq d-1$ to get
\begin{equation*} \dim \ker \phi\, \geq\, 2dr-2r^2\, =\, 2r(d-r)\, \geq 2(d-1)\, >\, 0\, , \end{equation*}
which implies $\det \phi = 0$.
\end{proof}

Before we can state and prove our main result, another technical lemma is necessary.

\begin{lemma} \label{pos semidef fix point} $\\$
Let $\phi\in\mathbf{Pt}$ be a positive, trace--preserving map whose spectrum contains $1$ with multiplicity strictly greater than $1$ (that is, whose fixed subspace verifies $\dim \eta_\phi>1$). Then $\phi$ admits a semipositive fixed point.
\end{lemma}

\begin{proof}
Let us call $\rho_0$ the positive fixed point of $\phi$ whose existence is guaranteed by Theorem~\ref{spect prop Pt}. If $\rho_0$ is semipositive we can immediately conclude. Otherwise, suppose $\rho_0>0$ and take another $X=X^\dag$ (independent from $\rho_0$) such that $\phi(X)=X$. Denoting by $\lambda_\mathrm{min}(Y)$ the minimum eigenvalue of the hermitian $Y$, define
\begin{multline*} A\ \equiv\ X\, -\, \lambda_\mathrm{min} (\rho_0^{-1/2} X \rho_0^{-1/2})\ \rho_0\ =\\
=\, \rho_0^{1/2}\, \left(\, \rho_0^{-1/2} X \rho_0^{-1/2} - \lambda_\mathrm{min} (\rho_0^{-1/2} X \rho_0^{-1/2})\ \mathds{1}\, \right)\, \rho_0^{1/2}\, . \end{multline*}
Then $\phi(A)=A$, and moreover $A$ must be semipositive, since
\begin{equation*} \lambda_\mathrm{min}\ \left(\ \rho_0^{-1/2} X \rho_0^{-1/2} - \lambda_\mathrm{min} (\rho_0^{-1/2} X \rho_0^{-1/2})\ \mathds{1}\ \right)\ =\ 0\, . \end{equation*}
\end{proof}

\subsection{A characterization theorem} \label{subs ES gen}
Now, we are in position to easily prove the following theorem, which is the main achievement of this section. We postpone our comments on the meaning of this result after its statement and proof. In what follows, $|\sigma_P(\phi)|$ will denote the number of peripheral eigenvalues of the channel $\phi$, counting multiplicities. Moreover, recall that we denote with $a_\phi(\lambda)$ the algebraic multiplicity of the eigenvalue $\lambda\in\sigma(\phi)$.

\begin{thm}[ES Channels with Nonzero Determinant] \label{ES det =/ 0} $\\$
Let $\phi\in\mathbf{CPt}_d$ be a quantum channel satisfying \mbox{$a_\phi(0) < 2(d-1)$} (in particular, $\det \phi \neq 0$ is a sufficient condition). Then the following are equivalent.

\begin{enumerate}

\item $\phi$ is entanglement--saving.

\item $\phi$ has a semipositive fixed point, or $| \sigma_P(\phi) | \geq 2$.

\item There exists $1\leq n\leq d$ such that $\phi^n$ has a semipositive fixed point.

\end{enumerate}

\end{thm}

\begin{proof}
\begin{itemize}

\item[]

\item[1 $\Rightarrow$ 2 :] This implication is true independently of the hypothesis $\det \phi \neq 0$. Suppose by contradiction that $\phi$ has a fixed point $\rho_0>0$ and \mbox{$\sigma_P(\phi)=\{ 1 \}$} (recall that $1$ always belongs to $\sigma(\phi)$, by Theorem~\ref{spect prop Pt}). Then it is not difficult to prove that \mbox{$\lim_{n\rightarrow\infty} \phi^n=D_{\rho_0}$}, where the depolarizing channel is defined by $D_{\rho_0}(X)\equiv\rho_0 \text{Tr} X$. This can be seen by noting that $\lim_{n\rightarrow\infty} \phi^n\equiv D$ is a (well--defined) channel whose output is always proportional to the positive eigenvector associated with the eigenvalue $1$ (because the others eigenvalues tend to zero when raised to arbitrary large powers). This can be written as $D(X)=\rho_0 f(X)$ for all $X$, where $f:\mathcal{M}(d;\mathds{C})\rightarrow \mathds{C}$ is a linear functional, and we can normalize $\text{Tr}\, \rho_0 = 1$. Applying the trace--preserving condition yields immediately $f(X)\equiv \text{Tr}\, X$.

Next, let us observe that the Choi--Jamiolkowski isomorphism is linear (in particular, continuous), and so 
\begin{equation*} \lim_{n\rightarrow\infty} R_{\phi^n}\ =\ R_{\lim_{n\rightarrow\infty}\phi^n} \ =\ R_{D_{\rho_0}}\ =\ \rho_0\otimes\frac{\mathds{1}}{d}\ . \end{equation*}
Since $\rho_0>0$, by Proposition~\ref{internal separable det 0} the limit of the sequence is internal to the set of separable states, and this implies $n(\phi)<\infty$, which is absurd.

\item[2 $\Rightarrow$ 3 :] Also this statement does not require the hypothesis $\det \phi \neq 0$. If $\phi$ admits a semipositive fixed point we can immediately conclude. Otherwise, thanks to Corollary~\ref{cor Wolf}, there exists \mbox{$1\leq n\leq d$} such that the spectrum of $\phi^n$ contains $1$ with multiplicity strictly greater than $1$. In that case, Lemma~\ref{pos semidef fix point} again guarantees the existence of a semipositive fixed point.

\item[3 $\Rightarrow$ 1 :] Here is where our restrictive hypothesis \mbox{$a_\phi(0)\leq 2(d-1)$} comes into play. Firstly, if for some $n=n_0$ the map $\phi^n$ has a semipositive fixed point, it is immediate to see that the same happens for each multiple of $n_0$, i.e. frequently in $n\in\mathds{N}$. Assume by contradiction that $n(\phi)<\infty$. Then there exists $N\in\mathds{N}$ such that $\phi^N$ is entanglement--breaking and has a semipositive fixed point. By Corollary~\ref{EB det 0 cor}, this would imply $\dim\ker(\phi^N)\geq 2(d-1)$. Since it is generally true that $\dim\ker(\phi^N)\leq a_\phi(0)$ (as can be easily seen using the Jordan decomposition~\eqref{Jordan}), we would deduce $a_\phi(0)\geq 2(d-1)$, which is absurd by hypothesis. If we assume the simplified hypothesis $\det \phi\neq 0$, then Corollary~\ref{EB det 0 cor} immediately gives the absurd equality $\det\phi^N=0$.

\end{itemize}
\end{proof}

Theorem~\ref{ES det =/ 0} completely solves the problem of finding an explicit characterization of the entanglement--saving property for a wide class of channels, i.e. those verifying $a_\phi(0) < 2(d-1)$ (or, simplifying, $\det\phi\neq 0$). Among its consequences, observe that we can immediately conclude that the set of ES channels has measure zero (since both the initial restricting condition and the other conditions given in the thesis define sets of measure zero).

From a geometrical point of view, we could say that Theorem~~\ref{ES det =/ 0} characterizes the ES set \emph{almost everywhere}, that is, apart from a set of measure zero. Actually, this does not ensures a priori that the statement of Theorem~\ref{ES det =/ 0} is useful (i.e. identifies a nonempty set of ES channels), because also the ES set has measure zero. As a matter of fact, \emph{there are many entanglement--saving channels whose determinant is equal to zero} (for instance). We will see that a large class of examples emerges in a natural way in the context of Section~\ref{sec AES}. 

However, we will see with the explicit example $d=2$ (see Subsection~\ref{subs ES q}) that what we have just proved indeed gives a useful characterization. Namely, the restriction $\det\phi\neq 0$ causes \emph{no loss of generality} for the simplest nontrivial case, i.e. those of qubit channels. In this sense, it is less severe than what we could imagine.

\subsection{Entanglement--saving qubit channels} \label{subs ES q}

Through this subsection, we explore some consequences of Theorem~\ref{ES det =/ 0}. In particular, we show that this result gives a complete characterization of the ES class in the case of channels acting on a single qubit. To proceed further, we need some simple lemmas. The first one discusses the consequences of the equation $\det \phi=0$ for a quantum qubit channel.

\begin{lemma} \label{det 0 q EB} $\\$
Let $\phi\in\mathbf{CPt}_2$ be a qubit channel such that $\,\det \phi =0$. Then $\phi$ is entanglement--breaking.
\end{lemma}

\begin{proof}
Let $\phi=(M,c)$ denote the Bloch representation~\eqref{Bloch act} for the channel $\phi$. By~\eqref{det q} and~\eqref{L matrix}, we have 
\begin{equation*} \det\phi\, =\, \det M\, =\, l_1(M)\,l_2(M)\,l_3(M)\, =\, 0\, .\end{equation*}
Then, at least one special singular value $l_i(M)$ of $M$ must be zero. Consequently, $\phi$ must necessarily have the sign--change property expressed in the fourth condition of Theorem~\ref{EB q}, and so it must be entanglement--breaking.
\end{proof}

What is shown in \eqref{Bloch act} is that every quantum channel acting on a two--dimensional system can be seen as an affine transformation sending the Bloch sphere into itself. Therefore, the image of the set of density matrices is represented by an ellipsoid contained in the Bloch sphere (we called it \emph{image ellipsoid}). Its principal axes' lengths are nothing but the singular values of $M$. The following result states some geometrically intuitive facts.

\begin{lemma} \label{norm infty 1 unital} $\\$
Let $(M,c)\in\mathbf{Pt}_2$ be a positive, trace--preserving, qubit map. Then we must have \mbox{$\|M\|_\infty\leq 1$}. Moreover, if \mbox{$\|M\|_\infty=1$} then $c=0$, i.e. the map is unital, and the image ellipsoid contains a pure state.
\end{lemma}

\begin{proof}
Consider a generic unit vector $n\in\mathds{R}^3$. Then the matrices $\mathds{1}\pm\vec{n}\cdot\vec{\sigma}$ are positive, because of~\eqref{Pauli spectrum}. Since $(M,c)$ is a positive map, $\mathds{1}\pm(\vec{c}+M\vec{n})\cdot\vec{\sigma}$ must be again positive operators, that is (see again~\eqref{Pauli spectrum})
\begin{equation*} |M(\pm n)+c|^2\ \leq\ 1\ . \end{equation*}
Taking one half the sum of these equations, one obtains
\begin{equation*} |M n|^2+|c|^2\ \leq\ 1\ . \end{equation*}
Since we can certainly choose $|Mn|=\|M\|_\infty$, we must have $\|M\|_\infty\leq 1$, where the equality sign can hold if and only if $c=0$. Moreover, we have already observed that the singular values of $M$ are the lengths of the principal axes of the image ellipsoid. Therefore, the image ellipsoid of a unital qubit channel with $\|M\|_\infty=1$ is necessarily tangent to the surface of the Bloch sphere, that is, it contains a pure state.
\end{proof}

Thanks to Lemma~\ref{det 0 q EB}, we can see that the (simplified) restriction $\det\phi\neq 0$ we considered in Theorem~\ref{ES det =/ 0} causes no loss of generality in the $d=2$ case. In fact, quantum channels with zero determinant are easily classified as entanglement--breaking (and so they are not entanglement--saving, of course). We are ready to use Theorem~\ref{ES det =/ 0} to obtain a classification of the ES qubit channels.

\begin{thm}[ES Qubit Channels] \label{ES q} $\\$
Let $\phi\in\mathbf{CPt}_2$ be a qubit channel. Then $\phi$ is entanglement--saving if and only if $\det \phi\neq 0$ and $\phi$ fixes or inverts a pure state. Here the ``\emph{inversion}'' is intended as the geometrical inversion $-\mathds{1}$ in the Bloch sphere. Observe that a map which inverts a pure state is necessarily unital.
\end{thm}

\begin{proof}
If $\det\phi\neq 0$ and $\phi$ fixes or inverts a pure state, then surely $\phi^2$ fixes one of them. In that case, Theorem~\ref{ES det =/ 0} guarantees the entanglement--saving property, because of the fact that a pure state is (as a density matrix) semipositive.

Let us turn our attention to the converse statement. If $\phi$ is entanglement--saving, then certainly $\det\phi\neq 0$ by Lemma~\ref{det 0 q EB}. Moreover, either $\phi$ has a semipositive fixed point (i.e. fixes a pure state), or \mbox{$|\sigma_P(\phi)|\geq2$} (again by Theorem~\ref{ES det =/ 0}). The first possibility gives us directly the thesis, so let us concern ourselves with the second one. If $M$ has an eigenvalue with unit modulus, then $\|M\|_\infty\geq 1$, and so Lemma~\ref{norm infty 1 unital} implies that $\phi$ is unital. This fact will be useful in a moment, but first observe that~\eqref{Wolf q} restricts the possible peripheral spectra to
\begin{equation*} \sigma_P(\phi)\ =\ \{1,1\},\ \{1,-1\},\ \{1,1,e^{i\theta},e^{-i\theta}\}\ . \end{equation*}
Since Lemma~\ref{pos semidef fix point} implies that $\phi$ must necessarily fix a pure state if \mbox{$\{1,1\}\subseteq\sigma_P(\phi)$}, let us examine the case $\sigma_P(\phi)=\{1,-1\}$. Recall that by Theorem~\ref{spect prop Pt} the $-1$ eigenvector can be chosen hermitian and traceless, i.e. of the form $n\cdot\vec{\sigma}$. Moreover, up to a simple rescaling, we can freely suppose $|n|=1$. In that case, using also the unitality, we obtain
\begin{equation*} \phi\ \left( \frac{\mathds{1}+n\cdot\vec{\sigma}}{2} \right)\ =\ \frac{\mathds{1}-n\cdot\vec{\sigma}}{2}\ .\end{equation*}
This is the same as saying that $\phi$ inverts the pure state $\frac{\mathds{1}+n\cdot\vec{\sigma}}{2}$ in the Bloch sphere.
\end{proof}

Theorem~\ref{ES q} gives us a geometrical characterization of the ES set for a single qubit. With this tool at hand, we can find an explicit parametrization of the ES set in the $d=2$ case. This is the content of the following theorem.

\begin{thm}[Explicit Form for ES Qubit Channels] \label{ES q explicit} $\\$
Let $\phi\in\mathbf{CPt}_2$ be a qubit channel represented in the Pauli basis (as in~\eqref{Bloch act}) by a matrix $M\in\mathcal{M}(3;\mathds{R})$ and a vector $c\in\mathds{R}^3$. Then $\phi$ is entanglement--saving if and only if one of the following two possibilities holds.
\begin{enumerate}

\item[]
\item There exist \mbox{$O\in\text{SO}(3)$}, \mbox{$\theta\in\mathds{R}$}, \mbox{$0<\lambda \leq 1$}, \mbox{$\lambda^2 \leq \mu \leq 1$}, \mbox{$\alpha\geq 0$} such that the complete positivity condition \mbox{$\alpha^2\leq(1-\mu)(\mu-\lambda^2)$} holds, and
\begin{multline} M\ =\ O\ M_+ (\lambda, \theta, \alpha, \mu)\ O^T\ \equiv\\
\equiv\ O\ \begin{pmatrix} \lambda \cos \theta & \lambda \sin \theta & \alpha \\ - \lambda \sin \theta & \lambda \cos \theta & 0 \\ 0 & 0 & \mu \end{pmatrix}\ O^T\ ,\label{ES q explicit 1M} \end{multline}
\begin{equation} c\ =\ O\ c_+ (\alpha, \mu)\ \equiv\ O\ \begin{pmatrix} -\alpha \\ 0 \\ 1-\mu \end{pmatrix}\ . \label{ES q explicit 1c}
\end{equation}

\item The channel is unital (that is, $c=0$), and there exist \mbox{$O\in\text{SO}(3)$}, \mbox{$\theta\in\mathds{R}$}, \mbox{$0<\lambda \leq 1$}, such that
\begin{multline} M\ =\ O\ M_- (\lambda, \theta)\ O^T\ \equiv\\
\equiv\ O\ \begin{pmatrix} \lambda \cos \theta & \lambda \sin \theta & 0 \\ \lambda \sin \theta & - \lambda \cos \theta & 0 \\ 0 & 0 & -1 \end{pmatrix}\ O^T\ . \label{ES q explicit 2} \end{multline} 

\end{enumerate}
\end{thm}

\begin{proof} $\\$
Thanks to Theorem~\ref{ES q}, we know that $\phi$ is entanglement--saving if and only if $\det\phi\neq 0$, and it fixes or inverts a pure state. Let us begin with the first possibility. 
In the following, recall the elementary property~\eqref{nUC}, which corresponds to the degree of freedom represented by $O$ in \eqref{ES q explicit 1M},~\eqref{ES q explicit 1c} and \eqref{ES q explicit 2}. Therefore, by applying if necessary an orthogonal matrix before the channel and its inverse after, we can suppose without loss of generality that the fixed point is $\Ket{0}\!\!\Bra{0}=\frac{\mathds{1}+e_3\cdot\vec{\sigma}}{2}$ (with \mbox{$e_3=(0,0,1)^T$}), i.e.
\begin{equation} Me_3+c\, =\, e_3\, . \label{q fix point} \end{equation}
The positivity condition which has to be imposed on $(M,c)$ can be written (exactly as in Lemma~\ref{norm infty 1 unital}) :
\begin{equation} |Mn+c|^2\, \leq\, 1\quad \forall\, n\in\mathds{R}^3\, :\ \ |n|=1\, . \label{q positivity 1} \end{equation}
Since the left--hand side of \eqref{q positivity 1} reaches its maximum at $n=e_3$, here its first--order variation must be zero. Then
\begin{multline*} 2\, \delta n^T\, M^T M\, e_3\, +\, 2\, \delta n^T\, M^T c\, \equiv\, 0 \quad \forall\, \delta n\perp e_3\quad\Rightarrow\\
\Rightarrow\quad M^T\, (M e_3+c)\, \propto\, e_3 . \end{multline*} 
Together with~\eqref{q fix point}, this gives $M^T e_3 = \mu e_3$ for some real $-1\leq\mu\leq 1$ (the restriction comes from the inequality $\|M\|_\infty\leq 1$ of Lemma~\ref{norm infty 1 unital}). This shows that there exist $m\in\mathcal{M}(2;\mathds{R})$ and $-1\leq\alpha,\beta\leq 1$ such that
\begin{equation*} M\ =\ \begin{pmatrix} m_{11} & m_{12} & \alpha \\ m_{21} & m_{22} & \beta  \\ 0 & 0 & \mu \end{pmatrix}\ ,\qquad c\ =\  \begin{pmatrix} -\alpha \\ -\beta \\ 1-\mu \end{pmatrix}\ . \end{equation*} 
It will be more simple to adopt the parametrization
\begin{equation*} m\ =\ \begin{pmatrix} s+d & a+b \\ a-b & s-d \end{pmatrix}\ . \end{equation*}  
Until now we have used only the positivity of $\phi$. In order to exploit the complete positivity, we have to impose that the Choi matrix \mbox{$R_\phi=(\phi\otimes I)(\Ket{\varepsilon}\!\!\Bra{\varepsilon})$} must be positive. In what follows, we will use for the bipartite system the computational basis sorted in lexicographical order, i.e. $\Ket{00},\Ket{01},\Ket{10},\Ket{11}$. With this convention, one has
\begin{equation*} R_\phi\ =\ \frac{1}{2}\ \begin{pmatrix} 1 & 0 & 0 & s + i b \\ 0 & 1-\mu & d - i a & - \alpha+i\beta \\ 0 & d + i a & 0 & 0 \\ s - i b & - \alpha-i\beta & 0 & \mu \end{pmatrix}\ . \end{equation*}
Take the $2\times 2$ principal minor composed of the second and third rows and columns. Then
\begin{multline*} R_\phi\geq 0\quad\Rightarrow\quad 0\ \leq\ \det\begin{pmatrix} 1-\mu & d-ia \\ d+ia & 0 \end{pmatrix}\ =\\
=\, -d^2-a^2\quad\Rightarrow\quad d=a=0\, . \end{multline*}
Let us call $s=\lambda \cos \theta$ and $b=\lambda \sin\theta$, with $\theta\in\mathds{R}$. Observe that $\lambda=0$ is prohibited by $\det\phi\neq 0$, and $\lambda>1$ would imply $\|M\|_\infty>1$. Since this would contradict Lemma~\ref{norm infty 1 unital}, we can require $0<\lambda\leq 1$. Then
\begin{equation*} R_\phi\ =\ \frac{1}{2}\ \begin{pmatrix} 1 & 0 & 0 & \lambda e^{i\theta} \\ 0 & 1-\mu & 0 & - \alpha+i\beta \\ 0 & 0 & 0 & 0 \\ \lambda e^{-i\theta} & - \alpha-i\beta & 0 & \mu \end{pmatrix}\ . \end{equation*}
Exploiting Silvester's criterion on principal minors (see p. 404 of~\cite{HJ1}), it is not difficult to prove that the positivity of this matrix is equivalent to
\begin{multline*} 0\ \leq\ \det\begin{pmatrix} 1 & 0 & \lambda e^{i\theta} \\ 0 & 1-\mu & -\alpha+i\beta \\ \lambda e^{-i\theta} & -\alpha-i\beta & \mu \end{pmatrix}\, =\\
=\, (1-\mu)(\mu-\lambda^2)\, - \alpha^2 - \beta^2\, . \end{multline*}

Until now, we have proved that, if $\det \phi \neq 0$ and \mbox{$\phi=(M,c)$} fixes a pure state, then there exists \mbox{$O\in\text{SO}(3)$}, $\theta\in\mathds{R}$, $0<\lambda \leq 1$, $\lambda^2 \leq \mu \leq 1$, $\alpha,\beta\in\mathds{R}$ satisfying the condition \mbox{$\alpha^2+\beta^2\leq(1-\mu)(\mu-\lambda^2)$} such that 
\begin{multline*} M\ =\ O\ \tilde{M}_+ (\lambda, \theta, \alpha, \beta, \mu)\ O^T\ \equiv\\
\equiv\ O\ \begin{pmatrix} \lambda \cos \theta & \lambda \sin \theta & \alpha \\ - \lambda \sin \theta & \lambda \cos \theta & \beta \\ 0 & 0 & \mu \end{pmatrix}\ O^T\, ,\end{multline*}
\begin{equation*} c\ =\ O\ \tilde{c}_+ (\alpha, \beta, \mu)\ \equiv\ O\ \begin{pmatrix} -\alpha \\ -\beta \\ 1-\mu \end{pmatrix}\, . \end{equation*}
To show that every such a pair $(\tilde{M}_+, \tilde{c}_+)$ is entanglement--saving, observe that
\begin{multline*} \left(\ \tilde{M}_+(\lambda, \theta, \alpha, \beta, \mu)\ ,\ \ \tilde{c}_+ (\alpha, \beta, \mu)\ \right)^n\ =\\
=\ \left(\ \tilde{M}_+(\lambda^n, n \theta, \alpha_n, \beta_n, \mu^n)\ ,\ \ \tilde{c}_+ (\alpha_n, \beta_n, \mu^n)\ \right)\ ,  \end{multline*}
with
\begin{equation*} \begin{pmatrix} \alpha_n \\ \beta_n \end{pmatrix}\ \equiv\ \left(\ \begin{pmatrix} \lambda\cos\theta & \lambda\sin\theta \\ -\lambda\sin\theta & \lambda\cos\theta \end{pmatrix}\ +\ \mu\ \mathds{1}\ \right)^n\ \begin{pmatrix} \alpha \\ \beta \end{pmatrix}\ . \end{equation*}

Therefore, by taking the partial transpose of $R_{\phi^n}$, one obtains

\begin{equation*} R_{\phi^n}^{T_B}\ =\ \frac{1}{2}\ \begin{pmatrix} 1 & 0 & 0 & 0 \\ 0 & 1-\mu^n & \lambda^n e^{in\theta} & - \alpha_n+i\beta_n \\ 0 & \lambda^n e^{-in\theta} & 0 & 0 \\ 0 & - \alpha_n-i\beta_n & 0 & \mu^n \end{pmatrix}\ .  \end{equation*}

The $2\times 2$ principal minor formed of the second and third rows and columns has negative determinant because $\lambda>0$, and this shows that $R_{\phi^n}^{T_B}$ can not be positive definite. Then the PPT criterion implies that $\phi^n$ can not be entanglement--breaking. Observe that it is possible to suppose $\beta=0$ and $\alpha\geq 0$ without compromising~\eqref{q fix point}, by means of the application of an appropriate rotation around $e_3$ before the channel and of its inverse after. In this way, one obtains~\eqref{ES q explicit 1M} and~\eqref{ES q explicit 1c}. This concludes the first part of the proof.

Now, let us concern ourselves with the second possibility. Suppose that $\det\phi\neq 0$ and that $\phi$ inverts a pure state. Proposition~\ref{norm infty 1 unital} shows that such a channel must be unital ($c=0$), since $\|M\|_\infty=1$. As in \eqref{q fix point}, we can suppose $Me_3=-e_3$. Moreover, to avoid $\|M\|_\infty>1$, the third row of $M$ can not contain any other nonzero element, i.e. there must exists
\begin{equation*} \begin{pmatrix} s+d & a+b \\ a-b & s-d \end{pmatrix}\ \in\ \mathcal{M}(2;\mathds{R}) \end{equation*}
such that
\begin{equation*} M\ =\ \begin{pmatrix} s+d & a+b & 0 \\ a-b & s-d & 0 \\ 0 & 0 & -1 \end{pmatrix}\ . \end{equation*}
Now, the corresponding Choi matrix becomes
\begin{equation*}  R_\phi\ =\ \frac{1}{2}\ \begin{pmatrix} 0 & 0 & 0 & s+ib \\ 0 & 1 & d-ia & 0 \\ 0 & d+ia & 1 & 0 \\ s-ib & 0 & 0 & 0 \end{pmatrix}\ . \end{equation*}
The positivity condition for such an object implies \mbox{$s=b=0$}, and \mbox{$d=\lambda\cos\theta,\ a=\lambda\sin\theta$}, again with \mbox{$0<\lambda\leq 1$} and \mbox{$\theta\in\mathds{R}$}. In order to show that every such a pair $(M_-(\lambda,\theta),0)$, with $0<\lambda\leq 1$, is entanglement--saving, it suffices to use the previous reasoning by observing that
\begin{equation*} M_-(\lambda,\theta)^{2n}\ =\ \tilde{M}_+(\lambda^{2n},0,0,0,1)\ . \end{equation*}

\end{proof}

Observe that the tho cases $(M_+,c_+)$ and $(M_-,0)$ are truly different only if $\lambda<1$ (in this case the spectra are indeed different). Conversely, if $\lambda=1$ it is always possible to bring back the second channel into the first form.

Thanks to this result, the set of ES qubit channels is essentially characterized (up to a unitary conjugation) by four parameters, that we called \mbox{$\lambda, \theta, \alpha, \mu$}. It could be useful to write once for all the action of the two maps
\begin{gather} \phi^+_{\lambda,\theta,\alpha,\mu}\ \equiv\ (\ M_+(\lambda,\theta, \alpha, \mu),\ c_+(\alpha, \mu)\ )\ , \label{phi+}\\
\phi^-_{\lambda,\theta}\ \equiv\ ( M_-(\lambda,\theta),\ 0) \label{phi-} \end{gather}
on a generic $2\times 2$ hermitian matrix:
\begin{equation} \phi^+_{\lambda,\theta,\alpha,\mu}\ \begin{pmatrix} a & b \\ b^* & c \end{pmatrix}\ \equiv\ \begin{pmatrix} a+(1-\mu)\ c & \lambda e^{i\theta}\ b - \alpha\ c \\ \lambda e^{-i\theta}\ b^* - \alpha\ c & \mu\ c \end{pmatrix}\ , \label{phi+ action} \end{equation}
\begin{equation} \phi^-_{\lambda,\theta}\ \begin{pmatrix} a & b \\ b^* & c \end{pmatrix}\ \equiv\ \begin{pmatrix} c & \lambda e^{-i\theta}\ b^* \\ \lambda e^{i\theta}\ b & a \end{pmatrix}\ . \label{phi- action}  \end{equation}

Observe that the two real parameters $\lambda,\theta$ can be joined together in order to form an unique complex parameter $z\equiv \lambda e^{i\theta}$ which satisfies $0<|z|\leq 1$ . 

In conclusion, let us examine a particular well--known set of ES channels in the following example.

\begin{ex}[Amplitude Damping Channels as ES] $\\$
The Amplitude Damping channels are a well--known set of qubit channels which reproduce the action of a spontaneous emission process, the system being coupled to a zero--temperature environment. They form a set parametrized by \mbox{$0\leq p\leq 1$} and defined by
\begin{equation*} AD_p\ \begin{pmatrix} a & b \\ b^* & c \end{pmatrix}\ \equiv\ \begin{pmatrix} a+(1-p)\ c & \sqrt{p}\ b \\ \sqrt{p}\ b^* & p\ c \end{pmatrix}\ . \end{equation*}
An easy inspection reveals that $AD_p\notin\mathbf{EBt}_2$ as soon as $p>0$. Moreover, since the composition rule is simply
\begin{equation} AD_p^n\,\equiv\, AD_{p^n}\, , \label{GAD comp n} \end{equation}
we deduce that for $p>0$ their direct $n$--index takes the value $+\infty$, that is, that they are entanglement--saving. As expected, a comparison with \eqref{phi+ action} shows that that
\begin{equation*} AD_p\ =\ \phi^+_{\lambda,\theta,\alpha,\mu} \quad \text{with} \quad \lambda=\sqrt{p},\ \theta=0,\ \alpha=0,\ \mu=p\ . \end{equation*} 
\end{ex}

\section{Asymptotically Entanglement--Saving Channels} \label{sec AES}

Through this section, we define and study an interesting subset of the entanglement--saving set, which we call \emph{asymptotically entanglement--saving} class. This further classification is based on the behaviour of the limit points of the sequence $\left( \phi^n \right)_{n\in\mathds{N}}$. In order to state clearly a well--posed definition, we need some technical preliminaries.

\subsection{Definition}

Recall that a limit point of a sequence is by definition the limit of one of its subsequences. Naturally, if a sequence admits more than one limit point, then it does not converge (e.g. the sequence $\left((-1)^n\right)_{n\in\mathds{N}}$ has the two limit points $+1$ and $-1$). With this concept at hand, we can now examine the structure of the limit points of a sequence $\left( \phi^n \right)_{n\in\mathds{N}}$.

\begin{lemma} \label{LPS} $\\$
Let $A\in\mathcal{M}(m;\mathds{C})$ be a complex square matrix. Then the sequence $\left( A^n \right)_{n\in\mathds{N}}$ has some (finite) limit points if and only if every eigenvalue $z\in\mathds{C}$ of $A$ verifies $|z|\leq 1$, and for each eigenvalue $z$ of modulus $1$ the corresponding Jordan blocks are trivial (that is, there are no off--diagonal elements). Every quantum channel $\phi\in\mathbf{CPt}_d$ has these properties.
\end{lemma}

\begin{proof}
It is enough to put $A$ in Jordan block form to see that the condition expressed in the thesis are necessary and sufficient in order to guarantee that the sequence $\left( A^n \right)_{n\in\mathds{N}}$ is not unbounded. Then, the Bolzano--Weierstrass theorem ensures that every bounded sequence in a finite--dimensional euclidean space (such as $\mathcal{M}(d;\mathds{C})$) admits a limit point. Theorem~\ref{spect prop Pt} states that the quantum channels enjoy these properties. Indeed, the boundness requirement for the sequence of powers (which are again $\mathbf{CPt}$ maps, and $\mathbf{CPt}_d$ is a compact set) is exactly the way we proved that spectral properties for quantum channels.
\end{proof}

Using this discussion as well as the proof of Theorem~\ref{spect prop Pt} as guidelines, the following Lemma should be quite obvious.

\begin{lemma} \label{pLPJf} $\\$
Let $A\in\mathcal{M}(m;\mathds{C})$ be a complex square matrix. As in~\eqref{Jordan}, write a Jordan decomposition for $A$, that is
\begin{equation} A\, =\, \sum_{k} (\lambda_k P_k + N_k)\ , \end{equation}
where the $\lambda_k$ are eigenvalues, the $P_k$ are projectors onto the generalized subspaces, and the $N_k$ are nilpotent applications. If the sequence $\left( A^n \right)_{n\in\mathds{N}}$ has some (finite) limit points, then
\begin{gather}
E_A\ \equiv\, \sum_{k:\, |\lambda_k|=1} P_k\ , \label{E A} \\
I_A\ \equiv\, \sum_{k:\, |\lambda_k|=1} \lambda_k^*\, P_k \label{I A}
\end{gather}
are two of them. Moreover, every limit point is diagonalizable in a Jordan basis for $A$, and has the form
\begin{equation} \sum_{k:\ |\lambda_k|=1 } z_k P_k\ ,\quad |z_k|\equiv 1\ \ \forall\ k \ . \label{limit diag} \end{equation} 
\end{lemma}

Now we are in position to prove a useful observation. The following Proposition explores the algebraic structure of the set of the limit points of a sequence $\left( A^n \right)_{n\in\mathds{N}}$.

\begin{prop}[Limit Points of the Powers of a Matrix as a Group] \label{LPG} $\\$
Let $A\in\mathcal{M}(m;\mathds{C})$ be a complex square matrix. Consider
\begin{equation} \mathcal{G}_A\,\equiv\,\{ \text{\ limit points of $\left(A^n\right)_{n\in\mathds{N}}$}\ \}\ . \end{equation}
Then $\mathcal{G}_A$, if not empty, is an abelian compact group with the standard operation of matrix multiplication.
\end{prop}

\begin{proof}
We will prove that $\mathcal{G}_A$ is closed under multiplication, possesses an identity element, is closed and limited as a set, and moreover that each element has an inverse.
\begin{itemize}
\item If $S,T\in\mathcal{G}_A$, then it must be $ST\in\mathcal{G}_A$. In fact, there exist subsequences $(k_n)_{n\in\mathds{N}},(h_n)_{n\in\mathds{N}}$ such that
\begin{equation*} S=\lim_{n\rightarrow \infty} A^{k_n}\ ,\quad T=\lim_{n\rightarrow \infty} A^{h_n}\ . \end{equation*}
But then
\begin{equation*} ST=\lim_{n\rightarrow \infty} A^{k_n+h_n}\ \in\ \mathcal{G}_A\ . \end{equation*}
\item Let us explicitly construct an identity element. If $S\in\mathcal{G}_A$, then Lemma~\ref{pLPJf} ensures that $S$ is diagonalizable in a Jordan basis for $A$. Therefore, with the same notation of~\eqref{E A}, equation~\eqref{limit diag} implies that $S\ E_A = S$. And so, $E_A$ is an identity element for $\mathcal{G}_A$.
\item Observe that $\mathcal{G}_A$ is closed as a set because of its definition. To see this, we consider a limit point of $\mathcal{G}_A$ and show that it actually belongs to $\mathcal{G}_A$ itself. Let $(S_k)_{k\in\mathds{N}}$ be a sequence of elements belonging to $\mathcal{G}_A$. Then for every $k$ there exists a sequence of powers of $A$ which converges to $S_k$. In other words, we can write
\begin{equation*} \tilde{S}=\lim_{k\rightarrow \infty} S_k\ ,\quad S_k=\lim_{n\rightarrow\infty} A^{h^{(k)}_n}\ . \end{equation*}
For each $k\geq 1$, define an integer $n_k$ such that
\begin{equation*} \big\| S_k - A^{h_{n_k}^{(k)}} \big\|_\infty\ \leq\ \frac{1}{k}\ \ . \end{equation*}
Then
\begin{equation*} \tilde{S}\ =\ \lim_{k\rightarrow\infty} S_k\ =\ \lim_{k\rightarrow\infty} A^{h_{n_k}^{(k)}}\ \in\ \mathcal{G}_A\ . \end{equation*}
\item A consequence of Lemma~\ref{pLPJf} is that $\mathcal{G}_A$ must be limited as a set. In fact, all matrices $S\in\mathcal{G}_A$ can be simultaneously diagonalized using a Jordan basis for $A$. In this basis our claim is obvious, because~\eqref{limit diag} guarantees that all the eigenvalues belong to the complex circumference $|z|=1$.
\item Let us prove that each generic $S\in\mathcal{G}_A$ has an inverse internal to $\mathcal{G}_A$ (and so, such an object must be unique). Observe that $S^k\in\mathcal{G}_A$ for each $k\in\mathds{N}$, and that Lemma~\ref{pLPJf} claims that $I_S$ (defined as in \eqref{I A}) is indeed a limit point of this sequence (which is limited because contained inside $\mathcal{G}_A$, and therefore has some limit points). Thanks to the property of closure of $\mathcal{G}_A$, we can deduce that $I_S\in\mathcal{G}_A$. Because of its definition, it must be $S\ I_S=E_A$, so $I_S$ is an inverse of $S$.
\end{itemize}
\end{proof}

Now, consider a quantum channel $\phi\in\mathbf{CPt}_d$. Thanks to Lemma~\ref{LPS}, Proposition~\ref{LPG} applies, and we can define the corresponding (non--empty) set of limit points $\mathcal{G}_\phi$. Since $\mathbf{CPt}_d$ is a closed set, it can immediately seen that $\mathcal{G}_\phi\subseteq\mathbf{CPt}_d$. Our immediate goal is to classify the entanglement--breaking properties of the elements belonging to $\mathcal{G}_\phi$. Fortunately, this task is quite easy; the answer is the content of the following proposition. \\

\begin{prop} \label{Gphi EB} $\\$
Let $\mathcal{G}_\phi$ be the non--empty set of limit points of the powers of a quantum channel $\phi\in\mathbf{CPt}$. There are only two possibilities:
\begin{equation*} \mathcal{G}_\phi\subseteq \mathbf{EBt}\quad\text{or}\quad \mathcal{G}_\phi\cap \mathbf{EBt}\, =\, \emptyset\ . \end{equation*}
\end{prop}

\begin{proof}
The only thing we have to prove is that $\mathcal{G}_\phi\subseteq \mathbf{EBt}$ if there exists $S_0 \in\mathcal{G}_\phi\cap \mathbf{EBt}$. Thanks to the group properties of $\mathcal{G}_\phi$, taken a generic $S\in\mathcal{G}_\phi$ we can certainly write
\begin{equation*} S\ =\ S_0\ (S_0^{-1} S)\ ,\quad S_0^{-1} S\in\mathcal{G}_\phi\subseteq \mathbf{CPt}\ . \end{equation*}
Recall the property~\eqref{EB propagates} of the entanglement--breaking channels. Since $S_0\in\mathbf{EBt}$, it must be also $S\in\mathbf{EBt}$.
\end{proof}

\begin{proof}[Alternative proof (more direct)]
These statements can be proved without invoking nor the particular group structure of $\mathcal{G}_\phi$, neither the other technical lemmas. The fact that $\mathcal{G}_\phi$ is not empty descends from the compactness of $\mathbf{CPt}_d$, as previously observed. Now, suppose that $S_0 \in\mathcal{G}_\phi\cap \mathbf{EBt}$. Then there exists a divergent sequence $(n_k)_{k\in\mathds{N}}$ such that \mbox{$\lim_{k\rightarrow \infty} \phi^{n_k}=S_0$}. Take another limit point $S=\lim_{h\rightarrow\infty}\phi^{m_h}\in\mathcal{G}_\phi$, and construct the divergent sequence $(h_k)_{k\in\mathds{N}}$, where \mbox{$h_k\equiv\min\,\{ h\in\mathds{N} :\, m_h\geq n_k\}$}. Then, for all $k\in\mathds{N}$ we can write 
\begin{equation} \phi^{m_{h_k}}=\phi^{m_{h_k}-n_k}\,\phi^{n_k}\, , \label{decompose} \end{equation}
with $m_{h_k}-n_k\geq 0$ by definition. Now, on one hand we have $\lim_{k\rightarrow\infty}\phi^{m_{h_k}}=S$ and $\lim_{k\rightarrow\infty}\phi^{n_k}=S_0$. On the other hand, the sequence $(\phi^{m_{h_k}-n_k})_{k\in\mathds{N}}$ is bounded (as it is composed by quantum channels), and so it must have at least one limit point $S_1\in\mathbf{CPt}$. Taking the limit of~\eqref{decompose} only on the subsequence that produces that limit point, we obtain $S=S_1 S_0$, and then from~\eqref{EB propagates} we infer again $S\in\mathbf{EBt}$. 
\end{proof}

Proposition~\ref{Gphi EB} makes a clear distinction between the two behaviours of the limit points. The following definition makes sense now; actually, it seems quite natural.

\begin{Def}[Asymptotically Entanglement--Saving Channels] \label{AES} $\\$
Let $\phi\in\mathbf{CPt}$ be a quantum channel, and denote by $\mathcal{G}_\phi$ the non--empty set of limit points of the sequence $\left( \phi^n \right)_{n\in\mathds{N}}$. If $\mathcal{G}_\phi\cap \mathbf{EBt} = \emptyset$ then $\phi$ is called \emph{asymptotically entanglement--saving} (AES).
\end{Def}

Remind that the set $\mathbf{EBt}$ is closed, and so its complement in $\mathbf{CPt}$ is open. Consequently, every AES channel is also ES, but the converse is not necessarily true. Moreover, consider a limit point $S\in\mathcal{G}_\phi$ of the sequence of powers of an AES channel $\phi$. We know (by definition) that $S$ is not entanglement--breaking. But not only: since $\mathcal{G}_\phi$ is a closed set (see Proposition~\ref{LPG}), it turns out that $S$ itself must be an AES channel! This means that
\begin{equation} \text{$\phi$ is AES}\ \Leftrightarrow\ \text{$\mathcal{G}_\phi$ is entirely composed of AES channels.} \label{lim again AES}\end{equation}
This observation leads to an elementary way to construct examples of ES channels with determinant equal to zero, that is, ES channels that might elude the classification provided by Theorem~\ref{ES det =/ 0}. Indeed, consider $S\in\mathcal{G}_\phi$, $\phi$ being a non--unitary AES channel. On one hand,~\eqref{lim again AES} ensures that $S$ is AES and so ES. On the other hand, Theorem~\ref{unitary} and Theorem~\ref{spect prop Pt} impose the constraint $|\det\phi|<1$, which in turn implies that 
\begin{equation*} \det S=\lim_{n\rightarrow\infty}\det(\phi^n)=\lim_{n\rightarrow\infty}(\det\phi)^n=0\ . \end{equation*}
An explicit example of this construction is given in Example~\ref{ex simple AES}.

Another useful characterization of the AES set descends directly from its definition. Since Proposition~\ref{Gphi EB} guarantees that either all the limit points are entanglement--breaking or the channel is AES, we can restrict ourselves to check only a particularly simple limit point, for instance the $E_\phi$ defined through~\eqref{E phi}.
\begin{equation} \text{$\phi$ is AES}\ \Leftrightarrow\ E_\phi\notin\mathbf{EBt}\, . \label{AES E phi}\end{equation}

What is the physical meaning of Definition~\ref{AES}? These AES channels represent a particularly innocuous kind of entanglement--saving noise, in the following sense. It can happen that a quantum channel never breaks completely the entanglement, even if it is applied many times; this is the entanglement--saving property. However, these repeated application can reduce the quantum correlations to an arbitrary low value, and destroy them \emph{only in the limit}. An asymptotic entanglement--saving channels does not exhibit such a behaviour. Instead, \emph{a finite amount of entanglement is present also in the limit}. In other words, suppose that Alice makes sure that only an AES noise is acting on her half of the global system. Then she is guaranteed that the bipartite system will always contain a significant and concretely usable quantity of entanglement.

\subsection{Simple results}

This subsection presents some partial results concerning the AES channels. These results can be easily seen as corollaries of the general theory to be discussed later, but here we will deduce them independently. Then, it will be very instructive to see as they fit into the general scheme drawn by Theorem~\ref{AES CC}.

Firstly, let us present a simple link between the AES property and the size of the peripheral spectrum. Indeed, a large peripheral spectrum is enough to guarantee the asymptotic entanglement saving, as we will see in a moment. We start by recalling that the trace norm $\| A \|_1$ of a generic matrix $A$ can be bounded by the sum of the moduli of its eigenvalues $\{\lambda_i (A) \}$, that is
\begin{equation} \| A \|_1\ \geq\ \sum_i |\lambda_i (A) |\, . \label{bound tr norm} \end{equation}
We refer the reader interested in the proof to~\cite{HJ2}, p. 172. Moreover, let us recap the content of the so--called \emph{reshuffling separability criterion} (see the earlier works~\cite{Resh1}, \cite{Resh2} and \cite{Resh3}, or~\cite{GeometryQuantum} p. 355 for a good review with proofs). Let $\rho_{AB}$ be a separable state on a bipartite system $AB$, with $\dim \mathcal{H}_A=\dim\mathcal{H}_B=d$. Denote by $\phi\in\mathbf{CP}_d$ the unique linear map on states of $A$ associated to $\rho_{AB}$ via the Choi--Jamiolkowski isomorphism~\eqref{Choi state}), i.e. verifying $\, R_\phi=\rho_{AB}$. Considering $\phi$ as a $d^2\times d^2$ complex matrix, its trace norm can not exceed $d$. In short,
\begin{equation} R_\phi\in\mathcal{S}\quad \Rightarrow\quad \|\, \phi\, \|_1\, \leq\, d\, . \label{resh cr} \end{equation}

\begin{prop}[AES Channels and Peripheral Spectrum] \label{AES periph} $\\$
Let $\phi\in\mathbf{CPt}_d$ be a quantum channel. If $\phi$ has more than $d$ peripheral eigenvalues, i.e. $| \sigma_P(\phi) | > d$, then $\phi$ is AES.
\end{prop}

\begin{proof}
Suppose $| \sigma_P(\phi) | > d$. Then~\eqref{limit diag} guarantees that every limit point $S\in\mathcal{G}_\phi$ also verifies $| \sigma_P(S) | > d$. Thanks to~\eqref{bound tr norm}, this implies that
\begin{equation*} \| S \|_1\, \geq\, \sum_i |\lambda_i (S)|\, >\, d\, . \end{equation*}
By invoking~\eqref{resh cr}, we can see that this forbids $S\in\mathbf{EBt}_d$.
\end{proof}

The last partial result we are going to present concerns the AES qubit class. Indeed, if $d=2$ then Lemma~\ref{det 0 q EB} and Theorem~\ref{unitary} give us powerful tools to characterize the whole set of AES channels.

\begin{prop} \label{AES q} $\\$
A qubit channel is AES if and only if it is unitary.
\end{prop}

\begin{proof}
Obviously a unitary channel $\mathcal{U}$ is AES, because the limit points of $(\mathcal{U}^n)_{n\in\mathds{N}}$ are again unitary channels. Let us concern ourselves with the converse. Thanks to Theorem~\ref{unitary}, we have only to prove that every AES qubit channel $\phi$ satisfies the property $|\det\phi|=1$. Assume by contradiction that $|\det\phi|<1$. Then the elementary equality $\det(\phi^n)=(\det\phi)^n$ shows that the limit points of $(\phi^n)_{n\in\mathds{N}}$ have zero determinant. This is absurd, because Lemma~\ref{det 0 q EB} shows that this mere property implies that they are entanglement--breaking.
\end{proof}

\subsection{General characterization of AES channels}

Through this subsection, we will prove the theorem that completely solves the characterization problem for the AES channels. This section is where the very general theory presented in Section~\ref{sec advanced} (and in particular Theorem~\ref{phase sub}) comes heavily into play. We postpone the discussion of our results, as well as some clarifying examples, after the statement and proof of the central result.

\begin{thm}[Complete Characterization of AES Channels] \label{AES CC} $\\$
Let $\phi\in\mathbf{CPt}_d$ be a quantum channel. Then the following facts are equivalent:
\begin{enumerate}

\item $\phi$ is asymptotically entanglement--saving.

\item With the notation of Theorem~\ref{phase sub}, at least one of the $d_i^{(1)}$ associated to $\phi$ is strictly greater than $1$.

\item The von Neumann algebra $\eta_{\tilde{E}_\phi^\dag}$ (see~\eqref{eta E phi tilde}) is noncommutative.

\item $\phi$ admits at least two noncommuting phase points.

\end{enumerate}
\end{thm}

\begin{proof}
Through this proof, the notation will follow the one developed for Theorem~\ref{phase sub}.
\begin{itemize}

\item[$1.\Leftrightarrow 2.$] The equivalence~\eqref{AES E phi} states that $\phi$ is AES if and only if $E_\phi$ (as defined in~\eqref{E phi}) is not entanglement--breaking. Observe that~\eqref{E phi act} gives us an explicit expression for the action of $E_\phi$. So we could think that this is enough to decide whether $E_\phi$ is or not entanglement--breaking. The problem here is that this expression works only when the input matrix is of the form $X \oplus 0$, with respect to the decomposition $\mathds{C}^d=\mathcal{K}\oplus\mathcal{K}^\perp$. With the nomenclature of~\eqref{psi bar}, we can say that we only know the action of $\tilde{E}_\phi$, not the one of $E_\phi$.

But the amazing fact is that the difference between $E_\phi$ and $\tilde{E}_\phi$ does not matter when we only care about the entanglement--breaking behaviour. Indeed, 
\begin{equation} E_\phi\in\mathbf{EBt}\quad \Leftrightarrow\quad \tilde{E}_\phi\in\mathbf{EBt}\, . \label{original} \end{equation}
The implication $\Rightarrow$ is trivial, because $\tilde{E}_\phi$ is a restriction of $E_\phi$, and a restriction of an entanglement--breaking channels must be again entanglement--breaking. In order to prove the converse, take a generic input matrix $X$. Since $E_\phi^2=E_\phi$ because of the very definition~\eqref{E phi}, $E_\phi(X)$ must be a fixed point for $E_\phi$, that is its support must be contained in $\mathcal{K}$. Denoting by $\mathcal{P}:\mathcal{M}_d\rightarrow\mathcal{M}_r$ (with $r\equiv\dim\mathcal{K}$) the superoperator that restricts every input to $\mathcal{K}$, we can write $E_\phi(X)=\mathcal{P}E_\phi(X)\oplus 0$, or more generally $E_\phi=\mathcal{P}E_\phi\oplus 0$. But then
\begin{equation} E_\phi\, =\, E_\phi^2\, =\, E_\phi\left(\,\mathcal{P} E_\phi \oplus 0\, \right)\, =\, \tilde{E}_\phi \mathcal{P} E_\phi\, \oplus\, 0 \, . \label{original eq1}\end{equation}
This equality and a slight generalization of~\eqref{EB propagates} show that $E_\phi$ must be entanglement--breaking if so is $\tilde{E}_\phi$.

Thanks to~\eqref{original}, we know that $\phi$ is AES if and only if $\tilde{E}_\phi\notin\mathbf{EBt}$. Moreover,~\eqref{E phi act} shows that
\begin{equation} \tilde{E}_\phi (\cdot)\, =\, \bigoplus_i\ \text{Tr}_{i,2}\, [\, P_i (\cdot) P_i \, ] \otimes \rho_{i,2}\ . \label{E phi tilde act} \end{equation}
A straightforward inspection of~\eqref{E phi tilde act} immediately reveals that $\tilde{E}_\phi$ is entanglement--breaking if and only if \mbox{$d_i^{(1)}=1$} for all $i$. Indeed, if $d_i^{(1)}=1$ for all $i$ then \mbox{$\text{Tr}_{i,2}\, [\, P_i (\cdot) P_i \, ]$} becomes a scalar function, and~\eqref{E phi tilde act} takes the form~\eqref{Holevo form}, meaning that $\tilde{E}_\phi$ is entanglement--breaking. Conversely, if at least one of the $d_i^{(1)}$ is strictly greater than $1$, then there exists a nontrivial, preserved qudit ``inside'' the Hilbert space. To be more precise, consider an entangled state of the form $A=\Ket{\varepsilon}\!\!\Bra{\varepsilon}_{i,1; (i,1)'}\otimes\rho_{i,2}$, where 
\begin{equation*} \Ket{\varepsilon}_{i,1; (i,1)'}\, =\, \frac{1}{\sqrt{d_i^{(1)}}}\, \sum_{j=1}^{d_i^{(1)}}\, \Ket{j}_{i,1}\otimes \Ket{j}_{(i,1)'} \end{equation*}
is a partial maximally entangled state over two copies of the $(i,1)$--th subspace in accordance with the decomposition~\eqref{chi}. Since $\tilde{E}_\phi$ fixes that entangled state $A$, $\tilde{E}_\phi$ can not be entanglement--breaking.

\item[$2.\Leftrightarrow 3.$] Obvious consequence of~\eqref{eta E phi tilde}.

\item[$2.\Leftrightarrow 4.$] Observe that $\phi$ admits two noncommuting phase points if and only if at least two noncommuting elements belong to $\chi_\phi$. Equation~\eqref{chi} allows us to conclude that $\chi_\phi$ is noncommutative as a set if and only $d_i^{(1)}>1$ for some $i$.

\end{itemize}
\end{proof}

\begin{proof}[Alternative proof of $\,3.\Rightarrow 1.$] Let us present now an alternative, direct proof of the implication $\,3.\Rightarrow 1.$, which does not use the heavily technical theory of the phase subspace summarized in Theorem~\ref{phase sub}. Consider two hermitian, noncommuting matrices $X,Y\in\eta_{\tilde{E}_\phi^\dag}$, and define $Z\equiv X+iY\in \eta_{\tilde{E}_\phi^\dag}$. Since $[X,Y]\neq 0$, one can easily verify that $[Z,Z^\dag]\neq 0$. Then, Lemma~\ref{ZZ} allows us to construct the entangled state $Q(Z)$. Moreover, thanks to the \mbox{*--algebra} property of $\eta_{\tilde{E}_\phi^\dag}$, we know that $Z^\dag Z\in\eta_{\tilde{E}_\phi^\dag}$, so that indeed $\left( I\otimes \tilde{E}_\phi^\dag \right)\,\left(Q(Z)\right)\, =\, Q(Z)$. Since $Q(Z)$ is entangled, this imply that $\tilde{E}_\phi^\dag\notin\mathbf{EBu}$, that is (see~\eqref{EB adjoint}) $\tilde{E}_\phi\notin\mathbf{EBt}$. Equations~\eqref{AES E phi} and~\eqref{original} allows us to conclude that $\phi$ is AES.

\end{proof}

\subsection{Examples and discussion}

What is the intuitive meaning of Theorem~\ref{AES CC}? In our opinion, the most easily understandable requirement contained in Theorem~\ref{AES CC} is condition $4$. Remind that the only input states which survive in the limit of an infinite number of iterations are the phase points (by virtue of Theorem~\ref{spect prop Pt}). Then, condition $4$ says simply that the entanglement (the most genuine quantum mechanical property) can survive in that limit only if a noncommutation relation (the most basic feature of quantum mechanics) exists among the suriving states, i.e. the phase points. 

Before we give some general class of examples of AES channels, let us discuss how the simple results contained in Proposition~\ref{AES periph}
and Proposition~\ref{AES q} fit into the general scheme drawn by Theorem~\ref{AES CC}.

\begin{proof}[Proof of Proposition~\ref{AES q} (revisited).] The fact that the only AES qubit channels are the unitary evolutions can be easily proved now. In fact, in the qubit case the second condition of Theorem~\ref{AES CC} imposes that $\mathcal{K}=\mathds{C}^2$ is indeed the whole space, and that there is only one addend in the direct sum~\eqref{chi}. Moreover, also the tensor product decomposition in that addend must have a trivial second factor, that is $\rho_{i,2}=1$. Thanks to~\eqref{chi eq2}, this implies that $\phi$ is an unitary evolution.
\end{proof}

\begin{proof}[Proof of Proposition~\ref{AES periph} (revisited).] We have to prove that a peripheral spectrum strictly larger than $d$ (counting multiplicities) invariably denotes an AES behaviour. This can be easily regarded as a consequence of condition $4$ (or $3$) of Theorem~\ref{AES CC}. Indeed, suppose by contradiction that $|\sigma_P(\phi)|>d$ and that the associated (at least) $d+1$ phase points commute with each other. Then the linear span $\chi_\phi$ of the phase points has dimension at least $d+1$. Since $\phi$ is hermiticity--preserving, $\chi_\phi$ is a real subspace (see Theorem~\ref{spect prop Pt}), and so it contains $d+1$ linearly independent (and commuting) hermitian matrices. It is a well--known fact that commuting hermitian operators can be simultaneously diagonalized. We would obtain $d+1$ linearly independent, diagonal matrices. Since the diagonal is composed of only $d$ entries, this is clearly absurd.
\end{proof}

Finally, we present the simplest possible class of examples of AES channels.

\begin{ex}[Simple AES Channels] \label{ex simple AES} $\\$
Once again, we refer the reader to Theorem~\ref{phase sub} for the basic notation. Consider a Hilbert space $\mathds{C}^d=\mathcal{K}$ decomposed in the form
\begin{equation} \mathds{C}^d\, =\ \bigoplus_i\ \mathcal{K}_i\, =\ \bigoplus_i\ \mathcal{K}_i^{(1)}\otimes \mathcal{K}_i^{(2)}\ , \label{decompose} \end{equation}
with at least one of the $d_i^{(1)}\equiv\dim\mathcal{K}_i^{(1)}$ strictly greater than $1$. Call $P_i$ the orthogonal projector onto $\mathcal{K}_i$. Take unitary matrices $U_i$ acting on $\mathcal{K}_i^{(1)}$, density matrices $\rho_{i,2}$ defined on $\mathcal{K}_i^{(2)}$ and a permutation $\pi$ over the set of indices $i$ which exchanges only indices sharing the same dimension $d_i^{(1)}$. Then construct the channel $S$ given by the formula
\begin{equation} S(\cdot)\ \equiv\ \bigoplus_i\ U_i\, \left(\, \text{\emph{Tr}}_{\pi(i),2}\, [P_{\pi(i)}\, (\cdot)\, P_{\pi(i)}]\, \right)_{i,1}\, U_i^\dag\ \otimes\ \rho_{i,2}\ . \label{S act}\end{equation}

These channels enjoy the following properties.

\begin{itemize}

\item If at least one of the $d_i^{(1)}$ is strictly greater than $1$ (as we supposed), then $S$ is asymptotically entanglement--saving. This can be easily seen by observing that its phase subspace is indeed
\begin{equation*} \chi_S\, =\ \bigoplus_i\, \mathcal{M}_{d_i^{(1)}} \otimes \rho_{i,2}\ , \end{equation*}
and so it is not commutative (by Theorem~\ref{AES CC} this is enough to conclude).

\item These channels $S$ are the simplest AES channels in the sense that they have only eigenvalues either of unit modulus or equal to $0$. Indeed, it can be easily verified that the $S$ defined through~\eqref{S act} has $\sum_i {d_i^{(1)}}^2$ eigenvalues of unit modulus (see the construction reported in the proof of Theorem~\ref{phase sub} for more details). Moreover, it has two different kinds of zero eigenvalues: the first group is composed of $\sum_i {d_i^{(1)}}^2\, \left( {d_i^{(2)}}^2-1\right)$ zeros corresponding to eigenvectors of the form $\bigoplus_i a_{i,1}\otimes x_{i,2}$, where the $x_{i,2}$ satisfies \mbox{$\text{\emph{Tr}}\,x_{i,2}\equiv 0$} for all $i$. The second group corresponds to the $\sum_{i\neq j} d_i^{(1)}d_i^{(2)} d_j^{(1)} d_j^{(2)}$ independent hermitian input matrices having only off--diagonal elements with respect to the block decomposition~\eqref{decompose}. Observe that the total number of eigenvalues gives correctly $\left(\,\sum_i d_i^{(1)} d_i^{(2)}\, \right)^2$. The presence in the spectrum of only zero or unit complex eigenvalues by the way implies that $S$ itself is a limit point of the sequence $(S^n)_{n\in\mathds{N}}$ of its powers. This in turn shows that $S$ is an ES channel that might be impossible to classify according to Theorem~\ref{ES det =/ 0}, because it has a large number of zero eigenvalues. 

\end{itemize}

\end{ex}

\section{Conclusions}  \label{sec con} 

In the recent years enormous progresses have been put forward in the development of a brand new form of technology based on a clever use  of quantum mechanical systems. 
In particular, a lot has been learned on how to attenuate the detrimental effects of noise arising from the interaction of a system of interest with its surrounding environment, either by building better devices (hardware approach), or by exploiting complex coding procedures (software approach), yielding as an outcome an effectively noise whose intensity is milder than the original one.  
In view of all this, it makes sense to study more closely the properties of those processes which induce only small (yet nontrivial) perturbations on the system of interest. 
In order to characterize such a class of transformations, in the present work we focus on their ability of preserving entanglement of a bipartite quantum system when acting locally on one of the subparts. 
This yields to the introduction of three special set of superoperators, namely the universal preserving channels (i.e. those mapping which preserve any form of entanglement initially present in the system, no matter how weak it may be), the set of entanglement saving channels (i.e. maps which preserve entanglement of a maximally entangled state even after $n$ iterations, $n$ being an arbitrary integer), and the set of asymptotically entanglement saving channels (that make the entanglement survive even in the limit $n\rightarrow\infty$). First of all, we proved that the only universal entanglement preservers are the unitary evolutions. Then we found a partial characterization theorem for the entanglement saving channels (which in turn allows to fully understand the qubit case) and a complete characterization result holding for the asymptotically entanglement saving ones.

Most of the major open problems here concern the lack of a full understanding of the ES class. It is apparent that an ES channel has to produce very little noise in the system, but in what sense? Are there some easily verified criteria that can certify that a given channel is (or is not) entanglement saving? We are confident that a more complete knowledge of this kind of questions could reveal us some interesting phenomena appearing when the concatenation of quantum noises takes place.


\begin{thebibliography}{99}

\bibitem{ENTAN} R. Horodecki, P. Horodecki, M. Horodecki, and K. Horodecki, Rev. Mod. Phys. {\bf 81}, 865 (2009).

\bibitem{MONO1} W. Zurek, Rev. Mod. Phys. {\bf 75} , 715 (2003). 

\bibitem{REVIEW} A. S. Holevo and V. Giovannetti, Rep. Prog. Phys. {\bf 75},  046001 (2012). 

\bibitem{WolfQC} M. M. Wolf, \emph{Quantum Channels \& Operations}, lecture notes (2012).

\bibitem{HOLEVOBOOK} A. S. Holevo \emph{Quantum Systems, Channels, Information}  (de Gruyter Studies in Mathematical Physics, 2012).

\bibitem{HorodeckiShorRuskai} M. Horodecki, P. W. Shor, M. B. Ruskai, Rev. Math. Phys. {\bfseries 15}, 629 (2003).

\bibitem{V} A. De Pasquale, V. Giovannetti, Phys. Rev. A {\bfseries 86}, 052302 (2012).

\bibitem{LV} L. Lami, V. Giovannetti, arXiv:1411.2517 [quant-ph] (2014).

\bibitem{NC} M. A. Nielsen, I. L. Chuang, \emph{Quantum Computation and Quantum Information} (Cambridge University Press, Cambridge, 2000).

\bibitem{Kadison} R. Kadison, Ann. Math. {\bfseries 56}, 494-503 (1952).

\bibitem{Kadison Woronowicz} S. L. Woronowicz, Rep. Math. Phys. {\bfseries 10}, 165 (1976).

\bibitem{Bhatia} R. Bhatia, \emph{Positive Definite Matrices} (Princeton University Press, 2007).

\bibitem{Wigneroriginal} E. P. Wigner, Gruppentheorie (Friedrich Vieweg und Sohn, Braunschweig, Germany, 1931).

\bibitem{Wignerdirectproof} C. S. Sharma, D. F. Almeida, Ann. Phys. {\bfseries 64}, 300 (1990).





\bibitem{Holevoform} A. S. Holevo, Russian Math. Surveys {\bfseries 53}, 1295-1331 (1999).

\bibitem{KingRuskai} C. King, M. B. Ruskai, IEEE Trans. Info. Theory {\bfseries 47}, 192-209 (2001).

\bibitem{RuskaiEBqubit} M. B. Ruskai, Rev. Math. Phys. {\bfseries 15}, 643 (2003).

\bibitem{HJ1} R. A. Horn and C. R. Johnson, \emph{Matrix Analysis} (Cambridge University Press, Cambridge, 1990).

\bibitem{Lindblad} G. Lindblad, Lett. Math. Phys., {\bfseries 47}, 189-196 (1999).

\bibitem{Wolfinverse} M. M. Wolf, D. Perez--Garcia, Eprint arXiv:1005.4545v1 [quant-ph] (2010).



\bibitem{Vnonmaxent} A. De Pasquale, A. Mari, A. Porzio, V. Giovannetti, Phys. Rev. A {\bfseries 87}, 062307 (2013).

\bibitem{GurvitsBarnum} L. Gurvits, H. Barnum, Phys. Rev. A {\bfseries 68}, 042312  (2003).

\bibitem{PeresPPT} A. Peres, Phys. Rev. Lett. {\bfseries 77}, 1413 (1996).

\bibitem{HorodeckiPPT} M. Horodecki, P. Horodecki, R. Horodecki, Phys. Lett. A {\bfseries 223}, 1 (1996).

\bibitem{HJ2} R. A. Horn and C. R. Johnson, \emph{Topics in Matrix Analysis} (Cambridge University Press, Cambridge, 1994).

\bibitem{Resh1} K. Chen, L.-A. Wu, Quant. Inf. Comp. {\bfseries 3}, 193 (2003).

\bibitem{Resh2} O. Rudolph, J. Phys. A {\bfseries 36}, 5825 (2003).

\bibitem{Resh3} O. Rudolph, Phys. Rev. A {\bfseries 67}, 032312 (2003).

\bibitem{GeometryQuantum} I. Bengtsson and K. $\mathrm{\dot{Z}yczkowski}$, \emph{Geometry of Quantum States} (Cambridge University Press, Cambridge, 2006).

















\end{thebibliography}
\end{document}